\documentclass[12pt]{article}
\usepackage[title]{appendix}
\usepackage[edges]{forest}
\usepackage{bbm} 
\usepackage{amsfonts}
\usepackage{mathtools}
\usepackage{subcaption}
\usepackage{graphicx}
\usepackage{tabularx}
\usepackage[nomessages]{fp}% https://urldefense.com/v3/__http://ctan.org/pkg/fp__;!!Mih3wA!EkZNe97uCN1hll7cv4Fpp1j-KkDF9-cLts07c1tNEAU14GyajBfnb1Ce1PxanFEQVjEb2qxNJPZ8lAp8oDefwm8$ 

\usepackage{enumitem} 
\setlist{  
  listparindent=\parindent,
  parsep=0pt,
}

\usepackage{amsmath,amssymb,epsf,epsfig,amsthm,times,ifthen,epic,psfrag,datetime}
\usepackage{multirow,makecell,tikz}
\usepackage{pgfplots,xfp}
\usepackage{cite}     
\usepackage{fancyhdr}
\usepackage{setspace} 
\usepackage{lastpage} 

\usetikzlibrary{math}
\def\submitteddate{November 13, 2023} 
\def\reviseddate{April 9, 2024}

\renewcommand{\baselinestretch}{1}
\newcommand{\NOP}[1]{}
\newtheorem{theorem}              {Theorem}     [section]

\newtheorem{lemma}      [theorem] {Lemma}
\newtheorem{corollary}  [theorem] {Corollary}

\newtheorem{example}    [theorem] {Example}

\theoremstyle{definition}         
\newtheorem{definition} [theorem] {Definition}

\DeclarePairedDelimiter\ceil{\lceil}{\rceil}

\let\emptyset\varnothing

\newcommand{\Hide}[1]{}

\newcommand{\R}{\mathbb{R}}

\DeclareFontFamily{U}{matha}{\hyphenchar\font45}
\DeclareFontShape{U}{matha}{m}{n}{
      <5> <6> <7> <8> <9> <10> gen * matha
      <10.95> matha10 <12> <14.4> <17.28> <20.74> <24.88> matha12
      }{}
\DeclareSymbolFont{matha}{U}{matha}{m}{n}
\DeclareMathSymbol{\notdivides}{3}{matha}{"1F}
\DeclareMathSymbol{\divides}{3}{matha}{"17}

\usepackage{hyperref}
\hypersetup{
    colorlinks=true,
    citecolor=black,
    filecolor=black,
    linktoc=all,     %set to all if you want both sections and subsections linked
    linkcolor=blue,
    pdfpagemode=UseNone, % This prevents the bookmarks sidebar from popping up.
    urlcolor=black
}

\tikzset{
  iv/.style={
    draw,
    fill=orange!50,
    rectangle,
    minimum size=20pt,
    inner sep=0pt,
    text=black},
  ev/.style={
    draw,
    fill=green,
    rectangle,
    minimum size=20pt,
    inner sep=0pt,
    text=black}}

\begin{document}

\newcommand{\creationtime}{\today\ @ \currenttime}

\pagestyle{fancy}
\renewcommand{\headrulewidth}{0cm}
\chead{\footnotesize{Congero-Zeger}}
\rhead{\footnotesize{\reviseddate}}
%\rhead{\footnotesize{\submitteddate}}
%\rhead{\footnotesize{\today}}
\lhead{\footnotesize{\textit{Competitive advantage}}}
\cfoot{Page \arabic{page} of \pageref{LastPage}} % Need to include package ``lastpage''
\lfoot{\footnotesize{
Sections: 
\ref{sec:introduction}
\ref{sec:not-optimal}
\ref{sec:Lemmas}
\ref{sec:Huffman-over-SF}
\ref{sec:One-third-bounds}
\ref{sec:Shannon-Fano}
\ref{sec:small-codes}
\ref{sec:experimental-evidence}
\ref{sec:discussion}
\ref{sec:appendixA}
\nameref{references}
}}
%\rfoot{\footnotesize{Draft}}

% The line below adds a period after the Section number in section titles.
\makeatletter\renewcommand{\@seccntformat}[1]{\noindent {\csname the#1\endcsname}.\hspace{0.5em}}\makeatother

\renewcommand{\qedsymbol}{$\blacksquare$} % useful to control qed generated by \end{proof}

\newcommand{\Advantage}{\Delta}
\newcommand{\Alphabet}{S}

\setcounter{page}{1}

\title{Competitive Advantage of Huffman\\ and Shannon-Fano Codes
\thanks{
   \indent \textbf{S. Congero and K. Zeger} are with the 
  Department of Electrical and Computer Engineering, 
  University of California, San Diego, 
  La Jolla, CA 92093-0407 
  (scongero@ucsd.edu and ken@zeger.us).
}}

\author{Spencer Congero and Kenneth Zeger\\}
\date{
  \textit{
  IEEE Transactions on Information Theory\\
  Submitted: \submitteddate\\
%  Created: \creationtime \\
%  \Huge{Draft}
  Revised: \reviseddate\\
  }
}

\maketitle
\begin{abstract}
For any finite discrete source,
the competitive advantage of prefix code $C_1$ over prefix code $C_2$
is the probability $C_1$ produces a shorter codeword
than $C_2$,
minus the probability $C_2$ produces 
a shorter codeword than $C_1$.
For any source,
a prefix code is competitively optimal if it has
a nonnegative competitive advantage over all other prefix codes.
In 1991, Cover proved that
Huffman codes are competitively optimal 
for all dyadic sources,
namely sources whose symbol probabilities are negative integer powers of $2$.
We prove the following asymptotic converse:
As the source size grows,
the probability a Huffman code for a randomly chosen non-dyadic source
is competitively optimal converges to zero.
We also prove:
(i) For any non-dyadic source,
a Huffman code has
a positive competitive advantage over a Shannon-Fano code;
(ii) For any source, the competitive advantage 
of any prefix code over a Huffman code is strictly less than $\frac{1}{3}$;
(iii)
For each integer $n>3$,
there exists a source of size $n$ 
and some prefix code whose competitive advantage over a Huffman code is
arbitrarily close to $\frac{1}{3}$; and
(iv) For each positive integer $n$,
there exists a source of size $n$ 
and some prefix code whose competitive advantage over a Shannon-Fano code becomes
arbitrarily close to $1$ as $n\to\infty$.

\end{abstract}

\clearpage

[ \small{\textit{The Table of Contents below is provided for the reader's
 convenience and will be removed later}.} ]
\tableofcontents

\clearpage

\section{Introduction}
\label{sec:introduction}

In a probabilistic game where multiple players are each rated by a numerical score,
one way to designate a particular player $A$ as being 
superior among their fellow competitors
is by ``expected score optimality'',
where no other player $B$ can obtain a better score on average than the score of $A$.
A second way,
called ``competitive optimality'',
occurs if for every other player $B$,
the probability of $A$ scoring better than $B$ 
is at least
the probability of $B$ scoring better than $A$.
That is,
$A$ has a nonnegative ``competitive advantage'' 
\begin{align}
P(A\ \text{scores better than}\ B) - P(B\ \text{scores better than}\ A)
\end{align}
over all other players $B$
(where tie scores are ignored).
In this paper we obtain results about competitive advantage and optimality
when lossless source coding is viewed as a game,
source codes are the players,
the numerical score is the length of the codeword 
of a randomly chosen source symbol,
and shorter codeword lengths are better.

The most direct motivation for our study is its immediate connection
to the results of 
Cover~\cite{Cover-1991}
on the competitive optimality of Huffman codes for dyadic sources,
and follow-up papers by
Feder~\cite{Feder-1992}
and
Yamamoto and Itoh~\cite{Yamamoto-Itoh-1995},
all of which appeared in these
\textit{Transactions}.
Cover showed that Huffman codes were competitively optimal for dyadic sources
but was unable to extend his results to non-dyadic sources,
and instead allowed a (suboptimal) one-bit ``handicap'' 
in such cases in~\cite{Cover-1991},
as did Feder in~\cite{Feder-1992}.
Our analysis explains why they were were unable to extend the results
beyond dyadic sources.
Also, this material has become standard subject matter in the 
widely-used 
Cover-Thomas textbook~\cite{Cover-Thomas-book-2006}.

Our results directly connect to
the well-studied ``Game of Twenty Questions'',
which has become a standard example for most information theory classes,
also appearing in~\cite{Cover-Thomas-book-2006}.
The game consists of a player trying
to deduce which value of a discrete source was selected.
When viewed as a competition between two players,
one player ``wins'' (respectively, ``loses'') 
if that player asks fewer (respectively, more) questions that the opponent,
and otherwise ``ties''.
One analysis of the game seeks a strategy 
(analogous to a prefix code)
for a player that minimizes the average number
of questions needed to 
deduce which value of a discrete source was selected.
It has been shown that 
by representing question-asking strategies as prefix codes,
Huffman codes are optimal in this sense.
Using the same prefix code representation,
the results of our paper instead apply
to the question of
which of two competing strategies is more likely to win.

As another example related to the competitive analysis approach,
optimal stock portfolio selection 
has been studied in the field of financial investment,
both from the perspective of maximizing one's expected portfolio valuation
and alternatively from a competitive viewpoint
where one portfolio has a higher probability of being worth more
than a competing portfolio~\cite{Bell-Cover-1980}.

We also note an analogy from the game of golf,
where one scoring method called ``stroke play'' (or ``medal play'') is similar
to the average length of a variable length code,
whereas an alternative scoring method called ``match play''
is analogous to competitive advantage~\cite{Banks-golf-book}.

We next formalize some terminology and definitions and then explain
the history of the problem and our results.
Proof of many of the lemmas appear in the Appendix.

% ---------------------------------------------------------------------------

\subsection{Terminology}
\label{sec:terminology}

An \textit{alphabet} is
a finite set $\Alphabet$,
and a \textit{source} of size $n$ with alphabet $\Alphabet$ 
is a random variable $X$ 
taking on values in $\Alphabet$,
where $|\Alphabet|=n$ and 
we denote the probability that $X=y$ by
$P(y)$ for all $y\in\Alphabet$.
We also denote the probability of any subset $B\subseteq \Alphabet$  by
$P(B) = \displaystyle\sum_{y \in B} P(y)$.
A source is said to be \textit{dyadic} if 
$P(y)$ is a nonnegative integer power of $1/2$
for all $y\in\Alphabet$.

A \textit{code} 
for a given source $X$ is a mapping 
$C:\Alphabet \to \{0,1\}^*$
and for each $u\in S$, the binary string $C(u)$ 
is called a \textit{codeword} of $C$.
A \textit{prefix code} is a code
where no codeword is a prefix of any other codeword.

A \textit{code tree} for a prefix code $C$ is a rooted binary tree 
whose leaves correspond to the codewords of $C$;
specifically, the codeword associated with each leaf
is the binary word denoting the path from the root to the leaf.
The \textit{length} of a code tree node is its path length from the root.
The $r$th \textit{row} of a code tree is the set of nodes whose length is $r$,
and we will view a code tree's root as being on the top of the tree
with the tree growing downward.
For example, row $r$ of a code tree is ``higher'' in the tree than row $r+1$.
If $x$ and $y$ are nodes in a code tree,
then $x$ is a \textit{descendant} of $y$
if there is a downward 
path of length zero or more from $y$ to $x$.
Two nodes in a tree are called \textit{siblings} if they have the same parent.
In a code tree,
for any collection $A$ of nodes 
define $P(A)$ to be the probability
of the set of all leaf descendants of $A$ in the tree.

A (binary) \textit{Huffman tree} is a code tree constructed from a source
by recursively combining two smallest-probability nodes until only
one node with probability $1$ remains.
The initial source probabilities correspond to leaf nodes in the tree.
A \textit{Huffman code} for a given source is a mapping of source symbols to
binary words by assigning the source symbol corresponding to each leaf
in the Huffman tree to the binary word describing the path from the root
to that leaf.
A \textit{Shannon-Fano code} is a prefix code
such that for each $y\in\Alphabet$ 
the codeword associated with the source symbol $y$
has length 
$\lceil \log_2 \frac{1}{P(y)}\rceil$.

Given a source with alphabet $\Alphabet$ and a prefix code $C$,
for each $y\in \Alphabet$
the length of the binary codeword $C(y)$ is denoted $l_C(y)$.
Two codes $C_1$ and $C_2$ are \textit{length equivalent} if
$l_{C_1}(y) = l_{C_2}(y)$ for every source symbol $y\in\Alphabet$.
The \textit{average length} of a code $C$ for a source with alphabet $\Alphabet$ is
$\displaystyle\sum_{y\in\Alphabet} l_C(y) P(y)$.
A prefix code is \textit{expected length optimal} for a given source if no other prefix 
code achieves a smaller average codeword length for the source.

A prefix code is \textit{complete} if every non-root node in its code tree has a sibling,
or, equivalently, if every node has either zero or two children.
A code $C$ for a given source
is \textit{monotone} if
for any two nodes in the code tree of $C$,
we have $P(u)\ge P(v)$
whenever $l_C(u) < l_C(v)$.
Expected length optimal codes are always monotone
(see Lemma~\ref{lem:minlen-implies-monotone}).

% ---------------------------------------------------------------------------
\tikzset{every label/.style={xshift=0ex, text width=6ex, align=center, yshift=-6ex,
                             inner sep=1pt, font=\footnotesize, text=red}}
\begin{figure}[h]
\begin{center}
\scalebox{0.7}{
\begin{forest}
for tree={where n children={0}{ev}{iv},l+=8mm,
if n=1{edge label={node [midway, left] {0} } }{edge label={node [midway, right] {1} } },}
[$1$, label=a, baseline 
    [$\frac{5}{8}$, label=b] 
    [$\frac{3}{8}$, label=c  [$\frac{1}{4}$, label=d] 
                             [$\frac{1}{8}$, label=e] ]]
\end{forest}
}
\end{center}
\caption{A code tree for a prefix code of a source of size $3$.}
\label{fig:codetree}
\end{figure}
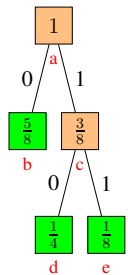
% ---------------------------------------------------------------------------

Figure~\ref{fig:codetree} illustrates some of the terminology.
The figure shows a code tree for a source of size $3$ with probabilities 
$\frac{5}{8}, \frac{1}{4}, \frac{1}{8}$
and corresponding codewords $0$, $10$, $11$.
The tree has root $a$,
the leaves are nodes $b,d,e$,
and the internal nodes are $a,c$.
The tree is complete since the non-root nodes come in sibling pairs,
namely $d,e$ and $b,c$.
This is a Huffman tree, since the two smallest nodes $d,e$ were first combined,
followed by the next two smallest nodes $b,c$ being combined.
The tree is monotone since
$1 = l(b) = l(c) < l(d)=l(e) = 2$ and $P(b), P(c) \ge P(d) , P(e)$.

The \textit{Kraft sum} of  a sequence of nonnegative integers
$l_1, \dots, l_k$ is 
$2^{-l_1} + \dots + 2^{-l_k}$.
We extend the definition of ``Kraft sum'' to also apply to sets of
source symbols or sets of nodes in a code tree as follows.
The Kraft sum $K(A)$ of any collection $A$ of leaves
is the Kraft sum of the corresponding sequence of codeword lengths.
The Kraft sum $K(A)$ of any collection $A$ of nodes
having no common leaf descendants
is the Kraft sum of 
the set of all leaf descendants of $A$ in the tree.
The Kraft sum of a collection of source symbols is the Kraft sum
of the corresponding leaves in a code tree.

For a subset $A$ of a source's alphabet
and for a Huffman code $H$ for the source,
we say the \textit{Huffman-Kraft sum}
of $A$ is
\begin{align}
K(A) = \sum_{x\in A} 2^{-l_H(x)},
\end{align}
where we reuse the ``$K$'' notation.
Note that the Huffman-Kraft sum of $A$ 
is the (usual) Kraft sum of the Huffman codeword
lengths of the symbols in $A$.

As previously mentioned,
one measure of the success of a source code that 
generally differs from expected length optimality is
called ``competitive optimality''
and has been considered in the form of 
one-on-one competitions between pairs of prefix codes to determine
in each competition
which code has the higher probability of producing a shorter codeword
for a given source.

For a given source with alphabet $\Alphabet$, and for any codes $C_1$ and $C_2$,
define 
\begin{align}
W &= \{ u \in \Alphabet : l_{C_1}(u) < l_{C_2}(u) \}\label{eq:7262b}\\
L &= \{ u \in \Alphabet : l_{C_1}(u) > l_{C_2}(u) \}\\
T &= \{ u \in \Alphabet : l_{C_1}(u) = l_{C_2}(u) \}. \label{eq:7262}
\end{align}
The sets $W$, $L$, and $T$
contain the source values of the 
\textit{wins}, \textit{losses}, and \textit{ties},
respectively, for code $C_1$
in a one-on-one competition against code $C_2$ to see which
produces shorter codewords for the given source.
Even though the notation for $W$, $L$, and $T$ does not explicitly
reference $C_1$ and $C_2$,
the codes involved will be clear from context.
Code $C_1$ is said to 
\textit{competitively dominate} 
code $C_2$ if $P(W) \ge P(L)$,
and
\textit{strictly competitively dominate}
code $C_2$ if $P(W) > P(L)$.
A prefix code is
\textit{competitively optimal} for a source of size $n$
if it competitively dominates all other
prefix codes for the same source.
The notion of competitive optimality dates back at least to 1980
in the field of financial investment~\cite{Bell-Cover-1980}.

\subsection{Prior results}
\label{sec:prior-results}

Huffman codes are known to be 
expected length optimal, monotone, and complete for every source 
(e.g., see~\cite{Cover-Thomas-book-2006}),
whereas
Shannon-Fano codes are monotone, but need not be expected length optimal nor complete.
Shannon-Fano codes are known to achieve the same average lengths
as Huffman codes whenever
the source is dyadic.
For non-dyadic sources,
Shannon-Fano codes always have larger average lengths than Huffman codes,
but nevertheless the average length of the Shannon-Fano code
(and thereby also the Huffman code) 
is less than
one bit larger than the source entropy~\cite{Cover-Thomas-book-2006}.
Also, monotonicity is known to be necessary for prefix code optimality,
as stated below.

\begin{lemma}[{e.g., \cite[p. 670]{Gallager-IT-1978}}]
For any source,
if a prefix code is expected length optimal,
then it is monotone.
\label{lem:minlen-implies-monotone}
\end{lemma}

The well-known Kraft inequality and its converse are stated next.
\begin{lemma}[Kraft, e.g.,~{\cite[Theorem 5.2.1]{Cover-Thomas-book-2006}}] % p. 107
The codeword lengths $l_1, \dots, l_n$ of any prefix code satisfy
$2^{-l_1} + \dots + 2^{-l_n} \le 1$.
Conversely, 
if a sequence $l_1, \dots, l_n$ of positive integers satisfies
$2^{-l_1} + \dots + 2^{-l_n} \le 1$,
then there exists a binary prefix code whose codeword lengths are
$l_1, \dots, l_n$.
\label{lem:Kraft-Inequality}
\end{lemma}
The following lemma
expresses an equality condition for 
Lemma~\ref{lem:Kraft-Inequality}.
The proof follows easily from the proof of Theorem 5.1.1
in the Cover-Thomas textbook
\cite{Cover-Thomas-book-2006}. % p. 107
A more general result can be found in 
Theorem 2.5.19 of
the Berstel-Perrin-Reutenauer textbook
\cite{Berstel-Perrin-Reutenauer-book-2009}. % p.81

\begin{lemma}
A prefix code is complete
if and only if
its Kraft sum equals $1$.
\label{lem:complete_KS_1}
\end{lemma}
We have defined sources and prefix codes to be finite throughout this paper,
but we note that
a complete infinite prefix code need not have Kraft sum $1$
(e.g., see \cite{Linder-Tarokh-Zeger}). % pg. 2027 col. 1

% ---------------------------------------------------------------------------
In 1991, Cover
proved that Shannon-Fano codes
are competitively optimal for dyadic sources. 
Since Huffman and Shannon-Fano codes are length equivalent for dyadic sources
(via ~\cite[Theorem 5.3.1]{Cover-Thomas-book-2006}), % p. 111, using equality condition
Huffman codes are also competitively optimal in this case,
as reworded below.
\begin{theorem}[{Cover~\cite[Theorem 2]{Cover-1991}}] % p. 173
Huffman codes are competitively optimal for all dyadic sources.
\label{thm:Cover-Huffman-competitively-optimal-for-dyadic-sources}
\end{theorem}

Cover's proof also showed that for all dyadic sources the 
Huffman code strictly competitively dominates all other 
(i.e., not length equivalent) prefix codes.
Additionally, he showed that for all non-dyadic sources
Shannon-Fano codes competitively dominate all other prefix codes
if an extra one-bit penalty 
is assessed to the non-Shannon-Fano code during each symbol encoding.
Specifically, this ``penalty''
favorably treats ties as if they were wins for the Shannon-Fano code
and treats one-bit losses for the Shannon-Fano code as if they were ties,
thus giving the Shannon-Fano code a one-bit handicap in betting parlance.
As dyadic sources are rare among all sources,
it is natural to ask whether 
Cover's (non-handicapped) competitive optimality result extends 
in some way to non-dyadic sources.

In 1992, Feder~\cite{Feder-1992}
showed that for all non-dyadic sources
Huffman codes competitively dominate all other prefix codes
if an extra one-bit penalty 
is assessed to the non-Huffman code during each symbol encoding.
This result is analogous to
(but uses a different proof technique from)
Cover's result for Shannon-Fano codes.

In 1995, Yamamoto and Itoh~\cite{Yamamoto-Itoh-1995} 
illustrated a non-dyadic source whose Huffman code 
was not competitively optimal.
They presented a prefix code for the source 
which strictly competitively dominates the Huffman code,
and whose win and loss probabilities are
$P(W)=0.5$ and $P(L)=0.4$, respectively.
Their example consists of a source of size $4$ and distinct prefix codes
$C_1$, $C_2$, and $C_3$,
such that 
$C_i$ competitively dominates $C_j$ whenever $(i,j) \in \{(1,2), (2,3), (3,1)\}$.
This example demonstrates that the relation of competitive dominance is not transitive.
One of these codes was the Huffman code,
so the Huffman code is not competitively optimal,
and no competitively optimal prefix code exists in this case.
They also provided a sufficient condition for a source not to have a
competitively optimal Huffman code.
Their results, however,
do not provide an indication of how many or few source codes
would have competitively optimal Huffman codes.

\begin{lemma}[Yamamoto and Itoh~{\cite[Theorem 3]{Yamamoto-Itoh-1995}}] % p. 2016
For any source,
every competitively optimal prefix code is expected length optimal.
\label{lem:Yamamoto-Itoh}
\end{lemma}

Note that the premise in 
Lemma~\ref{lem:Yamamoto-Itoh} 
might be vacuous
if no competitively optimal prefix code exists for a particular source,
in which case no conclusion can be drawn.

Although Huffman codes are always expected length optimal for any given source,
expected length optimal prefix codes need not be Huffman codes,
but they must be length equivalent to a Huffman code,
as stated in the next lemma.
\begin{lemma}[Manickman~{\cite{Manickman-2019}}]
For any source,
any expected length optimal prefix code is length equivalent to some Huffman code.
\label{lem:Manickman}
\end{lemma}

The following corollary follows immediately from
Lemma~\ref{lem:Yamamoto-Itoh} and
Lemma~\ref{lem:Manickman}.

\begin{corollary}
For any source,
a prefix code is
competitively optimal if and only if it is length equivalent to 
some competitively optimal Huffman code.
\label{cor:Yamamoto-Itoh}
\end{corollary}
One implication of the corollary is that
for any source,
if no Huffman code is competitively optimal,
then no prefix code is competitively optimal.

Two other works using competitive optimality include
Yamamoto and Yokoo~\cite{Yamamoto-Yokoo-1995},
who in 2001 studied competitive optimality for 
almost instantaneous variable-to-fixed length codes,
Bhatnagar~\cite{Bhatnagar-2021},
who in 2021 studied the use of competitive optimality for analyzing error 
probabilities in artificial intelligence systems.

The Cover-Thomas textbook~\cite[p. 131]{Cover-Thomas-book-2006} 
gives the following gambling interpretation of competitive optimality
and acknowledges the difficulty of 
mathematically analyzing Huffman codes:
\begin{quotation}
``To formalize the question of competitive optimality, 
consider the following
two-person zero-sum game: 
Two people are given a probability
distribution and are asked to design an 
instantaneous code for the distribution.
Then a source symbol is drawn from this distribution, and the
payoff to player A 
is $1$ or $-1$, 
depending on whether the codeword of
player A is shorter or longer than the codeword of player B. The payoff
is 0 for ties.

Dealing with Huffman code lengths is difficult, since there is no explicit
expression for the codeword lengths.''
\end{quotation}

To determine whether a prefix code is competitively optimal,
one must determine if the difference between the probabilities of
winning and losing a competition for shortest codeword length
is nonnegative for every possible opponent code.
When a prefix code $C_1$ is determined not to be competitively optimal,
at least one other code $C_2$ has a larger probability of winning
against $C_1$ than $C_1$ has against $C_2$.
It has not previously been known how large or small 
the difference between those probabilities can be
when, say, $C_1$ is a Huffman or Shannon-Fano code.
We next define terminology for that probability difference
and then in later sections derive an upper bound for it,
and describe the tightness of the bound.

\subsection{Competitive advantage}
\label{sec:Competitive-Advantage}

For a given source,
the \textit{competitive advantage} of  code $C_1$ over  code $C_2$
is the quantity
\begin{align}
\Advantage 
&= P(W) - P(L) .
\end{align}
Thus, 
$|\Advantage| \le 1$, and
$C_1$ competitively dominates $C_2$ if and only if $\Advantage\ge 0$.
Whereas competitive optimality indicates whether one code always dominates
other codes,
competitive advantage quantifies by how much one code
dominates over another code.
If the codes $C_1$ and $C_2$ are complete,
then their Kraft sums each equal $1$
by Lemma~\ref{lem:complete_KS_1},
and therefore neither code can always produce a shorter codeword
than the other code for every symbol.
In this case, $|\Advantage| < 1$.

If a source is dyadic, 
then there can be only one Huffman code
(up to length equivalence)
and this code is length equivalent to a
Shannon-Fano code.
However, 
for non-dyadic sources,
somewhat unusual circumstances can arise.
In what follows,
Example~\ref{ex:2-codes} illustrates a source where one Huffman code is
competitively optimal but another Huffman code for the same source is not,
and Example~\ref{ex:4-codes} illustrates a source with
two Huffman codes and a non-Huffman code 
that form a cycle of strict competitive domination,
as well as another non-Huffman code that is expected length optimal.

\begin{example}[Two Huffman codes]\ \\
A source with symbols $a,b,c,d$ and corresponding probabilities
$\frac{1}{3}, \frac{1}{3}, \frac{1}{6}, \frac{1}{6}$
has a Huffman code $H_1$ with codeword lengths
$2,2,2,2$,
and another Huffman code $H_2$ with lengths
$1,2,3,3$.
Each Huffman code has competitive advantage zero over the other.
However,
one can verify that $H_1$ is competitively optimal, whereas $H_2$ is not.
\label{ex:2-codes}
\end{example}

\begin{example}[Two Huffman codes and two other codes]\ \\
A source with symbols $a,b,c,d,e,f$ and corresponding probabilities
$\frac{1}{3}, \frac{1}{3}, \frac{1}{9}, \frac{1}{9}, \frac{1}{18}, \frac{1}{18}$
has a Huffman code $H_1$ with codeword lengths
$1,2,3,4,5,5$
and another Huffman code $H_2$ with lengths
$2,2,3,3,3,3$
(see Figure~\ref{fig:Four-codes}),
each with average codeword length equal to $\frac{7}{3}$.
Two other trees for codes $C_1$ and $C_2$ are shown,
which are relabeled versions of the trees for $H_2$ and $H_1$, respectively.
$H_1$ strictly competitively dominates $H_2$
since its competitive advantage is
$P(a) - P(d,e,f) = \frac{1}{9}$,     % 1/3 - 1/9 - 1/18 - 1/18
while $C_1$ strictly competitively dominates $H_1$
since its competitive advantage is
$P(b,c) - P(a) = \frac{1}{9}$,  % 1/3 + 1/9 + 1/9 + 1/18 - 1/3
and $H_2$ strictly competitively dominates $C_1$
since its competitive advantage is
$P(a,d,e,f) - P(b,c)  = \frac{1}{9}$. % 1/3 + 1/9 + 1/18 + 1/18 - 1/3 -  1/9 
That is, $H_1$, $H_2$, and $C_1$ form a cycle,
which illustrates the non-transitivity of strict competitive dominance.
Also, $C_2$ is an expected length optimal code,
but it is non-Huffman since nodes $e$ and $f$ were not merged.

% ---------------------------------------------------------------------------

\tikzset{every label/.style={xshift=0ex, text width=6ex, align=center, yshift=-6ex,
                             inner sep=1pt, font=\footnotesize, text=red}}

\begin{figure}[h]
\hspace*{0.6cm} $H_1$ \hspace*{3.6cm} $H_2$ \hspace*{3.4cm} $C_1$ \hspace*{4.2cm} $C_2$ 
\begin{center}
\scalebox{0.7}{
\begin{forest}
for tree={where n children={0}{ev}{iv},l+=8mm,
if n=1{edge label={node [midway, left] {0} } }{edge label={node [midway, right] {1} } },}
[$1$, baseline 
    [$\frac{1}{3}$, label=a] 
    [$\frac{2}{3}$ [$\frac{1}{3}$, label=b] 
        [$\frac{1}{3}$ [$\frac{1}{9}$, label=c] 
           [$\frac{2}{9}$ [$\frac{1}{9}$, label=d] 
              [$\frac{1}{9}$ [$\frac{1}{18}$, label=e] 
                 [$\frac{1}{18}$, label=f]]]]]]]
\end{forest}
\hspace*{1.0cm}
\begin{forest}
for tree={where n children={0}{ev}{iv},l+=8mm,
if n=1{edge label={node [midway, left] {0} } }{edge label={node [midway, right] {1} } },}
[$1$, baseline
    [$\frac{2}{3}$ [$\frac{1}{3}$, label=a] [$\frac{1}{3}$, label=b]]
    [$\frac{1}{3}$ [$\frac{2}{9}$ [$\frac{1}{9}$, label=c][$\frac{1}{9}$, label=d]]
       [$\frac{1}{9}$ [$\frac{1}{18}$, label=e] [$\frac{1}{18}$, label=f]]]]]
\end{forest}
\hspace*{1.0cm}
\begin{forest}
for tree={where n children={0}{ev}{iv},l+=8mm,
if n=1{edge label={node [midway, left] {0} } }{edge label={node [midway, right] {1} } },}
[$1$, baseline
    [$\frac{1}{3}$, label=b] 
    [$\frac{2}{3}$ [$\frac{1}{9}$, label=c] 
        [$\frac{5}{9}$ [$\frac{1}{3}$, label=a] 
           [$\frac{4}{9}$ [$\frac{1}{9}$, label=d] 
              [$\frac{7}{18}$ [$\frac{1}{18}$, label=e] 
                 [$\frac{1}{18}$, label=f]]]]]]]
\end{forest}
\hspace*{1.0cm}
\begin{forest}
for tree={where n children={0}{ev}{iv},l+=8mm,
if n=1{edge label={node [midway, left] {0} } }{edge label={node [midway, right] {1} } },}
[$1$, baseline 
     [$\frac{2}{3}$ [$\frac{1}{3}$, label=a] [$\frac{1}{3}$, label=b]]
     [$\frac{1}{3}$ [$\frac{1}{6}$ [$\frac{1}{9}$, label=c][$\frac{1}{18}$, label=e]]
       [$\frac{1}{6}$ [$\frac{1}{9}$, label=d] [$\frac{1}{18}$, label=f]]]]]
\end{forest}
%
% ---------------------------------------------------------------------------
}
\end{center}
\caption{
Code trees of four prefix codes for a source of size $6$.
}
\label{fig:Four-codes}
\end{figure}
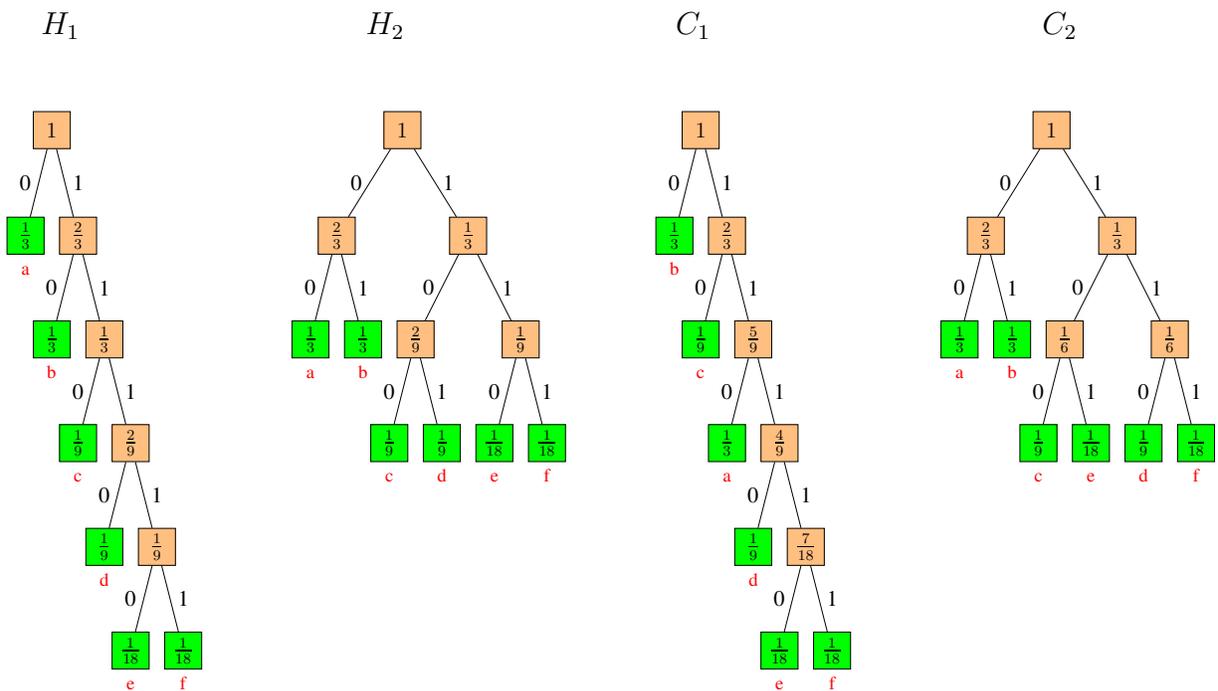

\label{ex:4-codes}
\end{example}

\subsection{Summary of contributions}
\label{sec:Summary-of-Contributions}

The following questions have not previously been answered in the literature:
(i) How likely or unlikely is it
for a Huffman code to be competitively optimal,
allowing for non-dyadic sources?;
(ii) If a Huffman code is not competitively optimal for a particular source, 
how large can the competitive advantage $\Advantage$ of another code be
over the Huffman code?;
(iii) Do Huffman codes always (perhaps, strictly) 
competitively dominate Shannon-Fano codes for
non-dyadic sources,
and how large can the competitive advantage $\Advantage$ of another code be
over the Shannon-Fano code?

One approach we exploit to answer the first question is to choose a source
``at random'' and seek the probability that a resulting
Huffman code is competitively optimal.
For the second and third questions,
one can seek the best possible upper bound on the
competitive advantage over a Huffman code
or over a Shannon-Fano code.

In this paper we address these questions by proving the following main results:
(1) The probability that a
Huffman code for a rather generally chosen random source 
is competitively optimal
converges to zero as the source size grows
(Theorem~\ref{thm:Huffman-not-competitively-optimal-in-limit}),
and therefore the probability that competitively optimal prefix codes exist for such
sources also converges to zero
(Corollary~\ref{cor:CompetitivelyOptimalCodesAreRare});
(2) For all non-dyadic sources,
Huffman codes strictly competitively dominate
Shannon-Fano codes
(Theorem~\ref{thm:Huffman-beats-SF});
(3) For all non-dyadic sources, the competitive advantage $\Advantage$
of any code over a Huffman code is strictly less than $\frac{1}{3}$
(Theorem~\ref{thm:ub-one-third-Huffman});
(4) For each integer $n>3$,
there exists a non-dyadic source of size $n$ 
and some prefix code whose competitive advantage $\Advantage$
over a Huffman code is
arbitrarily close to $\frac{1}{3}$
(Theorem~\ref{thm:lb-one-third-Huffman});
(5) For each positive integer $n$,
there exists a non-dyadic source of size $n$ 
and a prefix code for the source
such that the competitive advantage $\Advantage$ 
of the code over a Shannon-Fano code for the source
becomes arbitrarily close to $1$ as $n\to\infty$,
and the average length of the code becomes arbitrarily close to one bit less than
the average length of the Shannon-Fano code as $n\to\infty$
(Theorem~\ref{thm:lb-one-Shannon-Fano}).

We also analyze ``small'' sources and show that
for all sources of size at most $3$,
Huffman codes are competitively optimal
(Theorem~\ref{thm:competitively-optimal-Huffman:n=3}),
and
sources of size $4$
for which Huffman codes are competitively optimal
can be characterized in terms of a certain convex polyhedral condition
(Theorem~\ref{thm:competitively-optimal-Huffman:n=4}).

Finally
we conducted computer simulations that drew a million random sources 
from a flat Dirichlet distribution for each source size 
up to $34$ source symbols
and determined whether the resulting Huffman code 
satisfies a sufficient condition
for competitively non-optimality.
Our numerical observations 
are given in Section~\ref{sec:experimental-evidence}
and indicate that the convergence proven
in Theorem~\ref{thm:Huffman-not-competitively-optimal-in-limit}
is very rapid for relatively small source sizes
(Figure~\ref{fig:1}).

\clearpage

% ---------------------------------------------------------------------------

\section{Asymptotic converse to Cover's theorem on competitive 
         optimality of Huffman codes}
\label{sec:not-optimal}

In this section, our goal is to analyze how likely it is that a source
will have a competitively optimal Huffman code
(and more generally any competitively optimal prefix code).
To address this question,
we need to make precise what ``how likely'' means in this context.
There are many possible ways to choose a source ``at random''.
By some means we wish to randomly obtain numbers
$p_1, \dots, p_n \in (0,1)$ whose sum equals $1$,
interpret them as probabilities,
and then determine whether a Huffman code constructed
from these probabilities is competitively optimal.

One commonly used  way to randomly obtain such probabilities is to sample
from a ``flat Dirichlet distribution''.
A \textit{flat Dirichlet distribution} of size $n$
has a uniform probability density on the $(n-1)$-dimensional simplex 
$\{ (x_1, \dots, x_n) \in [0,1]^n : x_1 + \dots + x_n = 1 \}$
embedded in $\R^n$
(e.g., see~\cite{Dirichlet-book}).
For example, when $n=3$, the point $(p_1, p_2, p_3)$ is chosen uniformly from
the ($2$-dimensional) triangle embedded in $\R^3$,
whose vertices are
$(0,0,1)$, $(0,1,0)$, and $(1,0,0)$.
Choosing a random source from a flat Dirichlet distribution
tends to be a natural approach since it treats all
coordinates equally and uniformly.
This method of random source selection was used,
for example,
to analyze average Huffman code rate redundancy in
\cite{Khosravifard-Saidi-Esmaeili-Gulliver-2007}
and
\cite{Rastegari-Khosravifard-Narimani-Gulliver-2014}.

Another way to randomly create a source of size $n$
is to choose $n$ positive i.i.d. samples $X_1, \dots, X_n$
according to some probability distribution,
form their sum $S_n = X_1 + \dots + X_n$,
and then construct the normalized sequence
$\frac{X_1}{S_n}, \dots, \frac{X_n}{S_n}$.
This technique specializes to a
flat Dirichlet distribution when
$X_1, \dots, X_n$ are i.i.d. exponentials with mean one,
as given in the following lemma.
\begin{lemma}[e.g., {\cite[Chapter 11, Theorem 4.1]{Devroye-book}}] % p. 594
If $X_1, \dots, X_n$ are i.i.d. exponential random variables with mean one
and $S_n = X_1 + \dots + X_n$,
then the joint distribution of $\frac{X_1}{S_n}, \dots, \frac{X_n}{S_n}$
is the same as that of a
flat Dirichlet distribution of size $n$. 
\label{lem:Flat-Dirichlet=iid-exponential}
\end{lemma}

Using this more general method 
for randomly generating a source of size $n$,
we prove in this section 
that if the density of the i.i.d. sequence is positive on at least
some interval $(0,\epsilon)$ with $\epsilon>0$,
then the probability a Huffman code 
is competitively optimal shrinks to $0$ as $n$ grows
(Theorem~\ref{thm:Huffman-not-competitively-optimal-in-limit}).
In other words, competitively optimal prefix codes become rare for large sources.
This result can be viewed as an asymptotic converse 
to Cover's theorem for dyadic sources
(i.e., Theorem~\ref{thm:Cover-Huffman-competitively-optimal-for-dyadic-sources}),
or from a pessimistic viewpoint,
it also indicates that Cover's result cannot be extended to many large sources
beyond the dyadic ones.

We also examine an important special case of 
Theorem~\ref{thm:Huffman-not-competitively-optimal-in-limit}
when the random variables $X_1, \dots, X_n$ are i.i.d. exponential.

Our Corollary~\ref{cor:Huffman-flat-Dirichlet}
notes that the result of
Theorem~\ref{thm:Huffman-not-competitively-optimal-in-limit}
is true when choosing a source at random from a 
flat Dirichlet distribution,
which we use in Section~\ref{sec:experimental-evidence} to
gather experimental evidence of the convergence rate as $n\to \infty$.

The next lemma shows that if one set of leaves of a Huffman code 
has a smaller Kraft sum than that of a second disjoint set of leaves,
then one can find a prefix code that competitively wins against
the Huffman code on the set of smaller Kraft sum,
loses on the set of larger Kraft sum,
and ties on all other leaves.

\begin{lemma}
For any source,
if $H$ is a Huffman code,
and $U$ and $V$ are disjoint subsets of the source's alphabet $\Alphabet$
whose Huffman-Kraft sums satisfy $K(U) < K(V)$,
then there exists a prefix code $C$
such that
\begin{align}
\{ x\in \Alphabet : l_C(x) < l_H(x) \} &= U\\
\{ x\in \Alphabet : l_C(x) > l_H(x) \} &= V.
\end{align}
\label{lem:KraftExistPrefixCode}
\end{lemma}

\newcommand{\ProofOfLemmaKraftExistPrefixCode}{
\begin{proof}[Proof of Lemma \ref{lem:KraftExistPrefixCode}]
Let $k$ be an integer such that
\begin{align}
K(U) \le (1 - 2^{-k})K(V)
\label{eq:8000}
\end{align}
and for each $x\in \Alphabet$, define the following integer:
\begin{align}
l(x) &= 
 \begin{cases}
   l_H(x) - 1 & \text{if}\ x \in U\\
   l_H(x) + k & \text{if}\ x \in V\\
   l_H(x)     & \text{if}\ x \not\in U \cup V.
 \end{cases}
\end{align}
Then
\begin{align}
\sum_{x\in \Alphabet} 2^{-l(x)}
&= \sum_{x \in U} 2 \cdot 2^{-l_H(x)}
 + \sum_{x \in V} 2^{-k} \cdot 2^{-l_H(x)}
 + \sum_{x \not\in U \cup V} 2^{-l_H(x)}  \\
&= 2K(U) + 2^{-k} K(V) + (1-K(U)-K(V)) \label{eq:8001} \\
&= 1 + K(U) - (1 - 2^{-k})K(V)  \\
&\le 1, \label{eq:8002}
\end{align}
where
\eqref{eq:8001} follows since 
the Kraft sum of all the Huffman codewords is $1$,
by Lemma~\ref{lem:complete_KS_1};
and
\eqref{eq:8002} follows from \eqref{eq:8000}.
Therefore,
the Kraft inequality 
(Lemma~\ref{lem:Kraft-Inequality})
implies that
there exists a prefix code $C$ 
whose codeword lengths are the values of $l(x)$ for all $x\in \Alphabet$.
That is,
$l_C(x) = l_H(x)-1$ for all $x \in U$;
$l_C(x) = l_H(x)+k$ for all $x \in V$;
and
$l_C(x) = l_H(x)$ for all $x \not\in U \cup V$,
and therefore
\begin{align}
\{ x \in \Alphabet : l_C(x) < l_H(x) \} &= U\\
\{ x \in \Alphabet : l_C(x) > l_H(x) \} &= V.
\end{align}
\end{proof}
}

%\ProofOfLemmaKraftExistPrefixCode

% ---------------------------------------------------------------------------

The following lemma gives a sufficient condition for competitive non-optimality
for a Huffman code.

\begin{lemma}
For any source and any Huffman code for the source,
if $U$ and $V$ are disjoint subsets of the source alphabet
whose Huffman-Kraft sums satisfy $K(U) < K(V)$ 
and whose probabilities satisfy $P(U) > P(V)$,
then the Huffman code is not competitively optimal for the source.
\label{lem:KraftSumProbDisagree}
\end{lemma}

\newcommand{\ProofOfLemmaKraftSumProbDisagree}{
\begin{proof}[Proof of Lemma \ref{lem:KraftSumProbDisagree}]
Let $\Alphabet$ denote the source alphabet
and let $H$ be a Huffman code.
By Lemma~\ref{lem:KraftExistPrefixCode},
since $K(U) < K(V)$,
there exists a prefix code $C$ such that
\begin{align}
\{ x\in \Alphabet : l_C(x) < l_H(x) \} &= U \\
\{ x\in \Alphabet : l_C(x) > l_H(x) \} &= V.
\end{align}
The competitive advantage of $C$ over the Huffman code $H$ is thus
$P(U) - P(V) > 0$,
so the Huffman code is not competitively optimal.
\end{proof}
}

%\ProofOfLemmaKraftSumProbDisagree

The following lemma is essentially given 
in~\cite[equation (4)]{Yamamoto-Itoh-1995},
but we give an alternate short proof here for completeness.

\begin{lemma}
For any source,
if $y$ and $y'$ are sibling nodes in a Huffman tree,
$z$ is a leaf node descendant of $y$,
and $P(z) < P(y) - P(y')$,
then the Huffman code is not competitively optimal for the source.
\label{lem:LeafProbLessThanDifference}
\end{lemma}

\newcommand{\ProofOfLemmaLeafProbLessThanDifference}{
\begin{proof}[Proof of Lemma \ref{lem:LeafProbLessThanDifference}]
We may assume $P(y) > P(y')$ since $P(z) > 0$.
Let $U$ and $V$ be the sets of leaf node descendants of $y$ and $y'$,
respectively.
Let $U' = U - \{z\}$.
Then,
$K(U') < K(U) = K(V)$
(since siblings have the same Kraft sum)
and
$P(U') = P(U) - P(z) > P(U) - (P(U) - P(V)) = P(V)$.
Then by Lemma~\ref{lem:KraftSumProbDisagree},
the Huffman code is not competitively optimal.
\end{proof}
}

%\ProofOfLemmaLeafProbLessThanDifference

% ---------------------------------------------------------------------------

A hypothetical converse to 
Theorem~\ref{thm:Cover-Huffman-competitively-optimal-for-dyadic-sources}
would say that Huffman codes for non-dyadic sources are never competitively optimal.
This is not true
(e.g., Theorem~\ref{thm:competitively-optimal-Huffman:n=3}),
but we are able to demonstrate that 
such a converse becomes probabilistically true
as the source size grows.

Our next theorem,
one of our main results,
demonstrates an asymptotic converse to 
Theorem~\ref{thm:Cover-Huffman-competitively-optimal-for-dyadic-sources},
%Cover's theorem 
by showing that competitively optimal Huffman codes become
rare for randomly chosen sources as their size grows.

\begin{theorem}
Let $\epsilon >0$ and
let $X_1,X_2,\ldots$ be an i.i.d. sequence of 
nonnegative random variables with a density 
which is positive on at least $(0,\epsilon)$,
and let $S_n = X_1 + \dots + X_n$.
Then the probability that a Huffman code for 
$\frac{X_1}{S_n}, \dots, \frac{X_n}{S_n}$
is competitively optimal converges to zero as $n \to \infty$.
\label{thm:Huffman-not-competitively-optimal-in-limit}
\end{theorem}

\begin{proof}
Let $F$ denote the distribution function for each $X_i$.
Let $\delta = \epsilon/24$.
For each
$k\in \{1,\dots,24\}$,
define the interval
\begin{align}
I_k &= (\epsilon - k\delta, \epsilon - (k-1)\delta ) . 
\end{align}
The intervals $I_1, \dots, I_{24}$ are disjoint and their union
lies in $(0,\epsilon)$.

\newcommand{\IndicatorFunction}{\mathbbm{1}}
Denote the indicator function for any $E \subseteq \mathbb{R}$ by
$\IndicatorFunction_E(x) = 1$ if $x \in E$ 
and $\IndicatorFunction_E(x) = 0$ otherwise.
For each $k \in \{1,\dots,24\}$,
since the binary random variables 
$\IndicatorFunction_{I_k}(X_i)$ are i.i.d. for all $i$, we have
\begin{align}
\frac{|\{ i \in \{1,\dots,n\} : X_i \in I_k \}|}{n}
&= \frac{1}{n} \sum_{i=1}^n \IndicatorFunction_{I_k}(X_i) \\
&\xrightarrow{\mathrm{a.s.}} E[\IndicatorFunction_{I_k}(X_1)] \label{eq:7300}\\
&= P(X_1 \in I_k)\\
&= F(\epsilon - (k-1)\delta) - F(\epsilon - k\delta) \\
&> 0, \label{eq:7288}
\end{align}
where 
\eqref{eq:7300} follows from the strong law of large numbers 
% page 80
(e.g.,~\cite[Theorem 6.1]{Billingsley-book-1986});
and
\eqref{eq:7288} follows since $F$ is increasing on $(0,\epsilon)$.
Let $A$ be the event that
all $24$ of the convergences in \eqref{eq:7300} occur.
The intersection of finitely many events with probability $1$
has probability $1$,
so $P(A) = 1$.

We will show that for any outcome $\omega \in A$,
there exists $N \ge 1$ such that
for all $n \ge N$,
any Huffman code for $\frac{X_1(\omega)}{S_n(\omega)}, \dots, \frac{X_n(\omega)}{S_n(\omega)}$
is not competitively optimal.
Suppose this has been done,
and for each $n \ge 1$
let $B_n$ be the event containing the outcomes $\omega$ such that
a Huffman code for
$\frac{X_1(\omega)}{S_n(\omega)}, \dots, \frac{X_n(\omega)}{S_n(\omega)}$
is competitively optimal.
If $\omega \in A$,
then $\omega$ appears in only finitely many $B_n$,
so 
$\omega \not\in \displaystyle\limsup_{n \to \infty} B_n
 = \bigcap_{n\ge 1}\bigcup_{j\ge n} B_j$.
Therefore
$\displaystyle\limsup_{n \to \infty} B_n \subseteq A^c$,
and so
the theorem is proved using
\begin{align}
\limsup_{n \to \infty} P(B_n)
&\le P(\limsup_{n \to \infty} B_n)
\le P(A^c)
= 0
\end{align}
where the first inequality follows from
Fatou's Lemma
(e.g.,~\cite[Theorem 4.1]{Billingsley-book-1986}). % page 48

We will now prove what we promised.
Let $\omega \in A$,
and consider the sequence $X_1(\omega),X_2(\omega),\dots$.
Let $H$ be a Huffman tree constructed from the $n$ probabilities
$\frac{X_1(\omega)}{S_n(\omega)}, \dots, \frac{X_n(\omega)}{S_n(\omega)}$.
Since $\omega \in A$,
the convergence in \eqref{eq:7300} holds for all $k \in \{1,\dots,24\}$,
so for each such $k$ there exists $N_k \ge 1$
such that for all $n \ge N_k$
we have
$|\{ i \in \{1,\dots,n\} : X_i(\omega) \in I_k \}| \ge 1$.

Let $n \ge \max(N_1,\dots,N_{24})$.
Then for each $k \in \{1,\dots,24\}$,
let $Y_k$ equal $X_i(\omega)/S_n(\omega)$ for some $X_i(\omega) \in I_k$.
That is, 
$Y_1 > Y_2 > \dots > Y_{24}$ are probabilities
corresponding to $24$ of the $n$ leaves in the Huffman tree $H$.
Note that for all $k$, 
we have $Y_k S_n(\omega) \in I_k$,
so
\begin{align}
\frac{\epsilon - k\delta}{S_n(\omega)} 
&< Y_k < \frac{\epsilon - (k-1)\delta}{S_n(\omega)} .
\label{eq:7301}
\end{align}

The sibling in $H$ of the leaf for $Y_{14}$ has probability at most $Y_{13}$,
for otherwise the leaf for $Y_{14}$ and its sibling
would not have been nodes with two of the smallest available probabilities 
for merging as required by the Huffman construction.
Then the probability $\widehat{Y}_{14}$ of the parent
of the leaf for $Y_{14}$ satisfies
(using \eqref{eq:7301})
\begin{align}
\widehat{Y}_{14}
&\le Y_{14}+Y_{13}
< \frac{\epsilon - 13\delta}{S_n(\omega)} + \frac{\epsilon - 12\delta}{S_n(\omega)}
= \frac{\epsilon - \delta}{S_n(\omega)}
< Y_1.
\end{align}
Then since the Huffman code $H$ is monotone
by Lemma~\ref{lem:minlen-implies-monotone},
the leaf for $Y_1$ appears on a row in $H$
that is at least as high
as the row on which the parent of $Y_{14}$ appears,
so the leaf for $Y_1$ appears on a row strictly higher than
the row on which the leaf for $Y_{14}$ appears.
Since $Y_1 > Y_2 > \dots > Y_{14}$,
monotonicity of Huffman codes shows
the row numbers on which the leaves for these $14$ probabilities appear
are non-decreasing,
so the conclusion from the previous sentence implies
the leaf for some probability in the sequence
$Y_1,Y_2,\dots,Y_{14}$
appears on a row strictly higher in $H$ than
the leaf for the next probability in the sequence.
Specifically,
there exists $m \in \{1,\dots,13\}$ such that
the leaf for $Y_m$ appears on a row $r$
strictly higher than the row $r'$ on which the leaf for $Y_{m+1}$ appears,
i.e.,
$r < r'$.

Similarly,
the sibling for the leaf for $Y_{21}$
has probability at most $Y_{20}$,
so the probability $\widehat{Y}_{21}$ of the parent
of the leaf for $Y_{21}$ satisfies
(using \eqref{eq:7301})
\begin{align}
\widehat{Y}_{21}
&\le Y_{21}+Y_{20}
< \frac{\epsilon - 20\delta}{S_n(\omega)} + \frac{\epsilon - 19\delta}{S_n(\omega)}
= \frac{\epsilon - 15\delta}{S_n(\omega)}
< Y_{14}
\le Y_{m+1}.
\end{align}
Then since Huffman codes are monotone,
the leaf for $Y_{m+1}$ appears on a row
that is at least as high
as the row on which the parent of $Y_{21}$ appears,
so the row $r'$ on which the leaf for $Y_{m+1}$ appears
is strictly higher than the row $r''$ on which the leaf for $Y_{21}$ appears,
i.e.,
$r' < r''$.

Let $U$ be the set consisting of the leaves for $Y_{m+1}$ and $Y_{21}$,
and let $V$ be the set consisting of the leaf for $Y_m$.
Then
\begin{align}
K(U)
&= 2^{-r'} + 2^{-r''}
< 2^{-(r'-1)}
\le 2^{-r}
= K(V),
\end{align}
and again using \eqref{eq:7301},
\begin{align}
P(V) - P(U)
&= Y_m - (Y_{m+1} + Y_{21})
 < \frac{\epsilon - (m-1)\delta}{S_n(\omega)}
  - \left( \frac{\epsilon - (m+1)\delta}{S_n(\omega)}
  + \frac{\epsilon - 21\delta}{S_n(\omega)}\right)
= \frac{-\delta}{S_n(\omega)}
< 0.
\end{align}
Then
Lemma~\ref{lem:KraftSumProbDisagree} implies
the Huffman code for $\frac{X_1(\omega)}{S_n(\omega)}, \dots, \frac{X_n(\omega)}{S_n(\omega)}$
is not competitively optimal,
which is what we wanted to show.
\end{proof}

The requirements for the density of the random variables chosen in the statement
of Theorem~\ref{thm:Huffman-not-competitively-optimal-in-limit}
are not very restrictive and are satisfied by a wide range of random variables.
Such densities include various versions of at least the following:
beta, 
chi-square,
Erlang,
exponential,
gamma,
Gumbel,
log-normal,
logistic,
Maxwell,
Pareto,
Rayleigh,
uniform,
and
Weibull.

The following corollary follows from
Lemma~\ref{lem:Flat-Dirichlet=iid-exponential}
and
Theorem~\ref{thm:Huffman-not-competitively-optimal-in-limit}.
It shows that competitively optimal Huffman codes become rare for large sources
selected uniformly at random from a simplex.

\begin{corollary}
Let $X_1, \ldots, X_n$ be the probabilities of a source
chosen randomly from a flat Dirichlet distribution in $\R^n$.
Then the probability that a Huffman code for this source
is competitively optimal converges to zero as $n \to \infty$.
\label{cor:Huffman-flat-Dirichlet}
\end{corollary}

The next corollary 
follows immediately from 
Corollary~\ref{cor:Yamamoto-Itoh} and
Theorem~\ref{thm:Huffman-not-competitively-optimal-in-limit}.
It shows that as the source size grows 
the existence of any
competitively optimal prefix codes becomes unlikely.

\begin{corollary}
Let $\epsilon>0$ and
let $X_1,X_2,\ldots$ be an i.i.d. sequence of 
nonnegative random variables with a density 
which is positive on at least $(0,\epsilon)$,
and let $S_n = X_1 + \dots + X_n$.
Then the probability that a competitively optimal prefix code exists 
for the source
$\frac{X_1}{S_n}, \dots, \frac{X_n}{S_n}$
converges to zero as $n \to \infty$.
\label{cor:CompetitivelyOptimalCodesAreRare}
\end{corollary}

% ---------------------------------------------------------------------------

\clearpage

\section{Lemmas for future sections}
\label{sec:Lemmas}

In this section,
we give a number of lemmas that are used in the remainder of this paper.

Huffman codes are expected length optimal and
Gallager~\cite[Theorem 1]{Gallager-IT-1978} 
characterized such codes in terms of a ``sibling property'',
which says 
the code is complete 
and the nodes in the code tree
can be listed in order of non-increasing probability with
each node being adjacent in the list to its sibling.
Specifically, Gallager showed that a prefix code is a Huffman code
if and only if it satisfies the sibling property.

As an illustration,
Figure~\ref{fig:sibling-property}
shows three expected length optimal code trees
$H_1, H_2, C$
for a source with probabilities
$\frac{1}{3}, \frac{1}{3}, \frac{1}{6}, \frac{1}{6}$.
The first two are Huffman trees but the third is not,
and all three trees have an average length of $2$.
The trees $H_1$ and $H_2$ satisfy the sibling property,
as seen by ordering their nodes as
$a,c,b,d,e,f,g$
and
$a,c,b,f,g,d,e$,
respectively.
In contrast,
any ordering of the nodes of tree $C$
must have $d$ and $f$ precede $e$ and $g$,
but then the sibling pairs $d,e$ and $f,g$
are not adjacent in the ordering.
Thus, $C$ does not have the sibling property.

% --------------------
\tikzset{every label/.style={xshift=0ex, text width=6ex, align=center, yshift=-6ex,
                             inner sep=1pt, font=\footnotesize, text=red}}
\begin{figure}[h]
\begin{center}
\begin{forest}
for tree={where n children={0}{ev}{iv},l+=8mm,
if n=1{edge label={node [midway, left] {0} } }{edge label={node [midway, right] {1} } },}
      [$1$, label=a, baseline
         [$\frac{1}{3}$, label=b]
         [$\frac{2}{3}$, label=c
             [$\frac{1}{3}$, label=d]
             [$\frac{1}{3}$, label=e
                 [$\frac{1}{6}$, label=f]
                 [$\frac{1}{6}$, label=g]]]]
\end{forest}
\hspace*{1.7cm}
\begin{forest}
for tree={where n children={0}{ev}{iv},l+=8mm,
if n=1{edge label={node [midway, left] {0} } }{edge label={node [midway, right] {1} } },}
      [$1$, label=a, baseline
         [$\frac{1}{3}$, label=b
           [$\frac{1}{6}$, label=d]
           [$\frac{1}{6}$, label=e]]
         [$\frac{2}{3}$, label=c
           [$\frac{1}{3}$, label=f]
           [$\frac{1}{3}$, label=g]]]
\end{forest}
\hspace*{1.7cm}
\begin{forest}
for tree={where n children={0}{ev}{iv},l+=8mm,
if n=1{edge label={node [midway, left] {0} } }{edge label={node [midway, right] {1} } },}
      [$1$, label=a, baseline
         [$\frac{1}{2}$, label=b
           [$\frac{1}{3}$, label=d]
           [$\frac{1}{6}$, label=e]]
         [$\frac{1}{2}$, label=c
           [$\frac{1}{3}$, label=f]
           [$\frac{1}{6}$, label=g]]]
\end{forest}
\end{center}

\hspace*{0.6cm} Huffman code $H_1$ 
\hspace*{2.4cm} Huffman code $H_2$
\hspace*{3.0cm} optimal code $C$
\caption{
Two Huffman trees and an optimal third code tree for a single source.
}
\label{fig:sibling-property}
\end{figure}
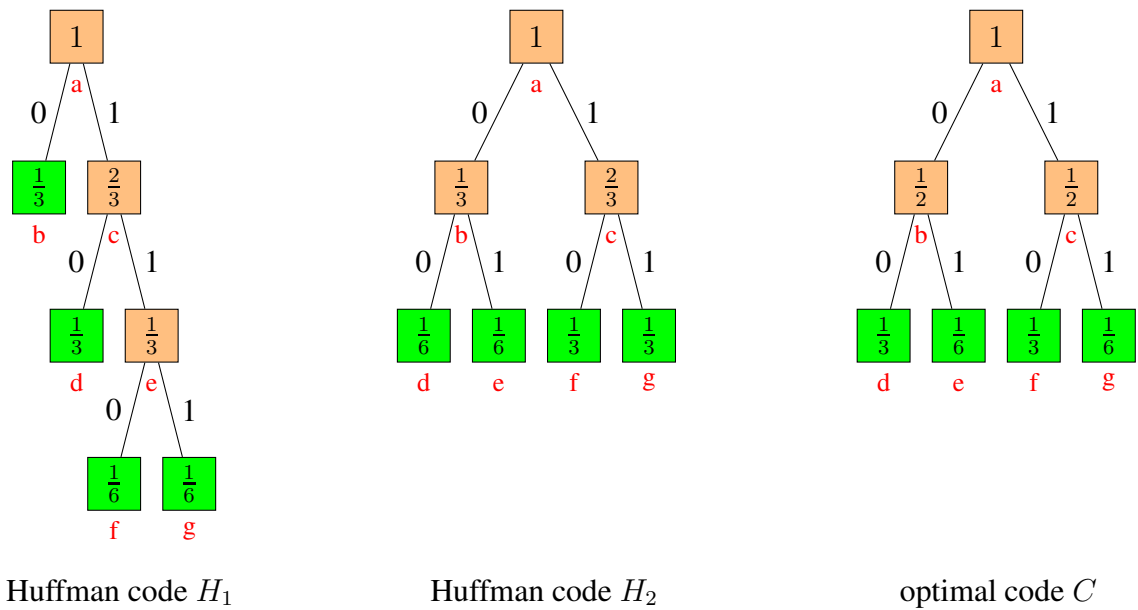
% --------------------

However, Huffman codes are not the only expected length optimal prefix codes.
Expected length optimal prefix codes are characterized as those which are 
length equivalent to some Huffman code for the source.
Below we note a second characterization of
expected length optimal prefix codes,
namely that such codes are 
precisely those which are complete and strongly monotone.
Strong monotonicity is a weaker condition than the sibling property,
and thus a broader class of prefix codes are strongly monotone.

\begin{definition}
Given a source with alphabet $\Alphabet$,
a prefix code $C$ is \textit{strongly monotone}
if for every $A, B \subseteq S$
with Kraft sums $2^{-i}$ and $2^{-j}$, respectively, 
for some integers $i$ and $j>i$,
we have
$P(A) \ge P(B)$.
\end{definition}

The strongly monotone property reduces to Gallager's monotone property
when the set $A$ consists of all leaf descendants of a single tree node,
and $B$ consists of all leaf descendants of a different tree node.

\begin{lemma}[{\cite[Theorem 2.4]{CoZe-strongly-monotone-ArXiv}}]
For any source and any prefix code $C$ for the source,
$C$ is expected length optimal
if and only if
$C$ is complete and strongly monotone.
\label{lem:strongly-monotone}
\end{lemma}

The following lemma notes several properties that 
are preserved under length equivalence
between prefix codes.

\begin{lemma}[{\cite{CoZe-strongly-monotone-ArXiv}}]
If two prefix codes are length equivalent,
then each of the following properties
holds for one code
if and only if
it holds for the other code:
\begin{itemize}
\item completeness
\item strong monotonicity
\item expected length optimality.
\end{itemize}
\label{lem:length_equivalent_properties}
\end{lemma}

For a given source and a Huffman code,
if $A$ is a nonempty proper subset of the source alphabet
whose Huffman-Kraft sum has binary expansion
$K(A) = 0.b_1 b_2 \ldots$,
then we define a \textit{Huffman-Kraft partition} of $A$ to be
any sequence $A_1, A_2, \ldots$ of disjoint 
(possible empty)
subsets of $A$ whose union is $A$ and
such that $K(A_i) = b_i 2^{-i}$ 
for each $i$.
If $b_k = 1$ and $b_i=0$ for all $i \ge k+1$, then it suffices
to specify the first $k$ sets $A_1, \dots, A_k$ in a Huffman-Kraft partition.

The next three lemmas are used in later sections,
and all make use of Huffman-Kraft partitions in their proofs.

\begin{lemma}
Every nonempty proper subset of a source's alphabet
has a Huffman-Kraft partition.
\label{lem:kraft_partition}
\end{lemma}

\newcommand{\ProofOfLemmaKraftPartition}{
\begin{proof}[Proof of Lemma \ref{lem:kraft_partition}]
Let $A$ be a nonempty proper subset of a source's alphabet and
let $0.a_1 a_2 \dots a_k$ be the binary expansion of the Huffman-Kraft sum $K(A)$,
where $a_k = 1$.

We use induction on $N = a_1 + \dots + a_k$.
Since $A$ is nonempty, $K(A)>0$, so $N \ge 1$.
First, suppose $N=1$.
Then $a_i = 0$ for all $i \in \{1,\dots,k-1\}$,
so $K(A) = 2^{-k}$.
Setting $A_k = A$
and $A_i = \emptyset$ for $i \in \{1,\dots,k-1\}$
gives a Huffman-Kraft partition of $A$.

Now let $m\ge 2$ and suppose the lemma holds for $N=m-1$.
We will next show the lemma must hold for $N=m$.
Let $X = \{ u \in A : K(u) \le 2^{-k} \}$.
Then $K(A - X)$ is an integer multiple of $2^{-k+1}$,
since for all $v \in A-X$
we have $K(v) = 2^{-k+i}$ for some $i \ge 1$.
Thus,
the binary expansion of $K(A-X)$ has $0$s
in all positions $i \ge k$.
Therefore, since $a_k=1$ and $K(A) = K(X)+K(A-X)$, 
the binary expansion of $K(X)$ has a $1$ in position $k$,
so $K(X) \ge 2^{-k}$.

Let $Y$ be a subset of $X$ whose Huffman-Kraft sum is minimum
among all subsets of $X$ with Huffman-Kraft sum at least $2^{-k}$.
We will show that $K(Y) = 2^{-k}$.
Suppose to the contrary that $K(Y) > 2^{-k}$.
Let $y \in Y$ be an element with minimum probability
among the elements of $Y$.
Since $y \in X$, we know $K(y) \le 2^{-k}$,
and also $K(y) = 2^{-i}$ for some integer $i$.
Therefore, $2^{-k}$ is an integer multiple of $K(y)$.
Also, $K(Y)$ is an integer multiple of $K(y)$,
since for all $u \in Y$
there exists an $i \ge 0$ such that $K(u) = 2^i K(y)$.
Therefore,
$K(Y) \ge 2^{-k} + K(y)$
since $K(Y) > 2^{-k}$.
But then $K(Y-\{y\}) = K(Y)-K(y) \ge 2^{-k}$
and $K(Y-\{y\}) = K(Y)-K(y) < K(Y)$,
so the Huffman-Kraft sum of $Y-\{y\}$ is at least $2^{-k}$ but is smaller
than that of $Y$,
contradicting the minimality assumption on $Y$.
Therefore, $K(Y) = 2^{-k}$.

Set $A_k = Y$.
Since $m \ge 2$, the set $A - A_k$ is nonempty,
and since $A$ is a proper subset of the source's alphabet, so is $A-A_k$.
Also, the Huffman-Kraft sum 
$K(A-A_k) = K(A) - K(A_k) = K(A) - 2^{-k}$
has exactly $m-1$ ones in its binary expansion.
Thus, the induction hypothesis implies that $A - A_k$ has a
Huffman-Kraft partition $A_1,\dots,A_{k-1}$.
Then $A_1,\dots, A_k$ is a Huffman-Kraft partition of $A$.
\end{proof}
}

%\ProofOfLemmaKraftPartition

\begin{lemma}
Given a source with alphabet $\Alphabet$,
suppose there is a subset $U \subseteq \Alphabet$
whose Huffman-Kraft sum $K(U)$
is an integer multiple of $2^{-i}$
for some integer $i \ge 0$.
If $A \subseteq U$
with $0 < K(A) < 2^{-j}$ for some integer $j \ge i$,
then there exists a subset $B \subseteq U - A$
with $K(A \cup B) = 2^{-j}$.
\label{lem:kraft_completion}
\end{lemma}

\newcommand{\ProofOfLemmaKraftCompletion}{
\begin{proof}[Proof of Lemma \ref{lem:kraft_completion}]
Since $K(A) > 0$
and $K(U - A) = K(U) - K(A) > 2^{-i} - 2^{-j} \ge 0$,
Lemma~\ref{lem:kraft_partition} shows
there exist Huffman-Kraft partitions
$A_1,A_2,\dots$ of $A$,
and $B_1,B_2,\dots$ of $U - A$.
Let 
\begin{align}
B = \displaystyle\bigcup_{k>j} B_k.
\end{align}
Then 
$K(B) = K(B_{j+1}) + K(B_{j+2}) + \dots 
      \le 2^{-(j+1)} + 2^{-(j+2)} + \dots = 2^{-j}$,
so $K(A \cup B) = K(A) + K(B) < 2^{-(j-1)}$.
Since $K(U) = K(A \cup B) + K(B_1 \cup \dots \cup B_j)$,
and both $K(U)$ and $K(B_1 \cup \dots \cup B_j)$
are integer multiples of $2^{-j}$,
it must be the case that
$K(A \cup B)$ is an integer multiple of $2^{-j}$ as well.
Then $0 < K(A \cup B) < 2^{-(j-1)}$ implies
$K(A \cup B) = 2^{-j}$.
\end{proof}
}

%\ProofOfLemmaKraftCompletion

\begin{lemma}
Given a source with alphabet $\Alphabet$,
suppose $A,B \subseteq \Alphabet$
with Huffman-Kraft sums satisfying
$2K(B) \le 2^{-i} \le K(A)$ for some integer $i$.
Then $P(A) \ge P(B)$,
with equality possible only if
$K(A) = 2K(B)$.
\label{lem:a_2b_inequality}
\end{lemma}

\newcommand{\ProofOfLemmaAtwoBInequality}{
\begin{proof}[Proof of Lemma \ref{lem:a_2b_inequality}]
If $B$ is empty,
then $P(B)=0$ and the result follows, so assume $B$ is nonempty.

If $K(B) < 2^{-(i+1)}$,
then by Lemma~\ref{lem:kraft_completion}
(where $A$, $B$, $U$, $i$, $j$ in Lemma~\ref{lem:kraft_completion}
 respectively correspond to $B$, $B'$, $S$, $0$, $i+1$ here)
there exists a subset $B' \subseteq \Alphabet - B$
with $K(B') = 2^{-(i+1)} - K(B)$.
Alternatively, if $K(B) = 2^{-(i+1)}$
then we will let $B' = \emptyset$.
In either case, $K(B \cup B') = 2^{-(i+1)}$.

By Lemma~\ref{lem:kraft_partition},
there exists a Huffman-Kraft partition
$A_1,A_2,\dots$ of $A$.
Let $k$ be the smallest integer such that $A_k \neq \emptyset$.
Since $K(A) \ge 2^{-i}$,
we have $k \le i$.

If $k < i$,
then since $0 < K(B \cup B') = 2^{-(i+1)} < 2^{-(k+1)}$,
Lemma~\ref{lem:kraft_completion} implies
(where $A$, $B$, $U$, $i$, $j$ in Lemma~\ref{lem:kraft_completion}
 respectively correspond to $B\cup B'$, $E$, $S$, $0$, $i+1$ here)
there exists a subset $E \subseteq \Alphabet - (B \cup B')$
such that $K(E) = 2^{-(k+1)} - K(B\cup B')$.
If $k = i$,
then instead let $E = \emptyset$.
In either case, $K(B \cup B' \cup E) = 2^{-(k+1)}$.

We have
\begin{align}
P(A)
&= P(A_k) + P(A - A_k)\\
&\ge P(A_k)\label{eq:7703}\\
&\ge P(B \cup B' \cup E) \label{eq:7701}\\
&= P(B) + P(B') + P(E) \label{eq:7702}\\
&\ge P(B). \label{eq:7700}
\end{align}
where
\eqref{eq:7701} follows 
the Huffman-Kraft sum equality
$K(A_k) = 2^{-k} = 2K(B \cup B' \cup E)$
since Huffman codes are strongly monotone by
Lemma~\ref{lem:strongly-monotone};
and
\eqref{eq:7702} follows since $B$, $B'$, and $E$ are disjoint. 

Now suppose $K(A) > 2K(B)$.
Then either $K(B) < 2^{-(i+1)}$ or $K(A) > 2^{-i}$.
In the first case,
we have $P(B') > 0$ since $K(B\cup B') = 2^{-(i+1)}$,
so the inequality in \eqref{eq:7700} becomes strict.
Now consider two subcases of the second case.
If $K(A) < 2^{-(i-1)}$,
then $i-1 < k \le i$,
so $k = i$ which implies
$K(A - A_k) = K(A - A_i) = K(A) - K(A_i) = K(A) - 2^{-i} > 0$,
and thus the inequality in \eqref{eq:7703} becomes strict.
Otherwise,
if $K(A) \ge 2^{-(i-1)}$,
then $k < i$,
and so $K(E) = 2^{-(k+1)} - K(B \cup B') = 2^{-(k+1)} - 2^{-(i+1)} > 0$.
Thus $P(E) > 0$,
and the inequality in \eqref{eq:7700} becomes strict.
\end{proof}
}

%\ProofOfLemmaAtwoBInequality

% ---------------------------------------------------------------------------

\clearpage

\section{Huffman codes competitively dominate Shannon-Fano codes}
\label{sec:Huffman-over-SF}

For dyadic sources,
there can be only one Huffman code
(up to length equivalence)
and such a code is length equivalent to a
Shannon-Fano code
~\cite[Theorem 5.3.1]{Cover-Thomas-book-2006}, % p. 111
and thus the competitive advantage 
of either code over the other code is $\Advantage = 0$.
In this section,
Theorem~\ref{thm:Huffman-beats-SF} shows 
that for all non-dyadic sources,
every Huffman code always strictly competitively dominates 
every Shannon-Fano code.

This result shows that Shannon-Fano codes are never competitively optimal
for non-dyadic sources.
Actually,
as the source size grows,
Huffman codes themselves 
are usually not competitively optimal either,
as shown in Section~\ref{sec:not-optimal}.

The following theorem shows that when a source is not dyadic,
every Huffman code always has a positive competitive advantage over
the Shannon-Fano code.

\newcommand{\LengthHuffman}{l_{\text{H}}(y)}
\newcommand{\LengthShannonFano}{l_{\text{SF}}(y)}

\begin{theorem}
Huffman codes strictly competitively dominate 
Shannon-Fano codes 
if and only if the source is not dyadic.
\label{thm:Huffman-beats-SF}
\end{theorem}

\begin{proof}
Suppose a non-dyadic source has alphabet $\Alphabet$
and a Huffman code $H$.
By Lemma~\ref{lem:strongly-monotone},
Huffman codes are strongly monotone,
so Huffman-Kraft sums obey the strong monotonicity condition.

Denote the Huffman and Shannon-Fano codeword lengths for each $y\in\Alphabet$
by $\LengthHuffman$ and $\LengthShannonFano$, respectively.
Let
$W = \{i\in\Alphabet : \LengthShannonFano  < \LengthHuffman \}$ 
and 
$L = \{i\in\Alphabet : \LengthShannonFano  > \LengthHuffman  \}$
and
$T = \Alphabet - (W \cup L)$,
as in \eqref{eq:7262b}--\eqref{eq:7262}.
It suffices to show $P(W) < P(L)$.

If $y \in W$ then 
\begin{align}
\log_2 \frac{1}{P(y)} \le \ceil*{\log_2 \frac{1}{P(y)}}
&= \LengthShannonFano  \le \LengthHuffman -1, 
\end{align}
and if $y \in L$ then
\begin{align}
\log_2 \frac{1}{P(y)} 
&> \ceil*{\log_2 \frac{1}{P(y)}} - 1
= \LengthShannonFano  -1 
\ge \LengthHuffman. 
\end{align}
Therefore, 
probabilities and Huffman-Kraft sums of wins, losses, and ties are bounded as
\begin{alignat}{2}
P(y) &\ge 2 \cdot 2^{- \LengthHuffman } = 2K(y) &&\ \ \ \ \text{if } y \in W    \\ 
P(y) &< 2^{- \LengthHuffman } = K(y) &&\ \ \ \ \text{if } y \in L    \\
K(y) &= 2^{- \LengthHuffman } \le P(y) < 2 \cdot 2^{- \LengthHuffman } = 2K(y) &&\ \ \ \ \text{if } y \in T. 
\end{alignat}
Thus,
for any nonempty subset $A \subseteq \Alphabet$,
\begin{alignat}{2}
P(A) &\ge 2K(A) &&\ \ \ \ \text{if } A \subseteq W   \label{eq:208} \\
P(A) &< K(A) &&\ \ \ \ \text{if } A \subseteq L   \label{eq:209} \\
K(A) &\le P(A) < 2K(A) &&\ \ \ \ \text{if } A \subseteq T, \label{eq:210}
\end{alignat}
since $P(A) = \displaystyle\sum_{y \in A} P(y)$
and $K(A) = \displaystyle\sum_{y \in A} K(y)$
for any subset $A \subseteq \Alphabet$.

Suppose that $L = \emptyset$.
Then from \eqref{eq:208} and \eqref{eq:210} we have
$P(y) \ge K(y)$ for all $y \in \Alphabet$.
Since $K(y)$ is an integer power of $1/2$
for all $y \in \Alphabet$,
there exists at least one element $y \in \Alphabet$
with $P(y) > K(y)$,
or else the source would be dyadic.
But then
\begin{align}
1
&= \sum_{y \in \Alphabet} P(y)
> \sum_{y \in \Alphabet} K(y)
= 1,
\end{align}
which is a contradiction.
Thus, in fact
$L \neq \emptyset$,
and therefore $P(L) > 0$.
This implies $P(W) < 1$.

If $W = \emptyset$
then $P(W) - P(L) < 0$,
and we are done.
Suppose $W \neq \emptyset$.
By Lemma~\ref{lem:kraft_partition},
there exist Huffman-Kraft partitions
$W_1,W_2,\dots$ and $L_1,L_2,\dots$
of $W$ and $L$, respectively.
Let $k$ be the smallest integer
such that
$W_k \neq \emptyset$.

If $k=1$,
then $W_1 \ne \emptyset$, so $K(W_1) = \frac{1}{2}$
and therefore we get the contradiction that
$1 > P(W_1) \ge 2 K(W_1) = 1$ by \eqref{eq:208}.
Thus $k \ge 2$.

We will use the following fact in the remainder of the proof.
If $A \subseteq \Alphabet$
and $K(A) = 2^{-(k-1)}$,
then
\begin{align}
P(A) &\ge P(W_k)\label{eq:7704}\\
&\ge 2K(W_k)\label{eq:7705}\\
&= 2^{-(k-1)}\label{eq:7706}\\
&= K(A) \label{eq:7707}
\end{align}
where
\eqref{eq:7704} follows from strong monotonicity;
\eqref{eq:7705} follows from \eqref{eq:208};
and
\eqref{eq:7706} follows from $P(W_k) = 2^{-k}$ since $W_k \ne \emptyset$.
Also, $L_{k-1} = \emptyset$, for otherwise
$K(L_{k-1}) = 2^{-(k-1)}$,
which by \eqref{eq:7707} would imply $P(L_{k-1}) \ge K(L_{k-1})$,
contradicting $P(L_{k-1}) < K(L_{k-1})$ by \eqref{eq:209},
as $L_{k-1} \subseteq L$.

Suppose $K(L) < 2^{-(k-2)}$.
Then since $L_{k-1} = \emptyset$,
at least the first $k-1$ bits in the binary expansion of $K(L)$ are zero, so
$K(L) < 2^{-(k-1)}$.
By Lemma~\ref{lem:kraft_completion},
(where $i$, $j$, $U$, $A$, $B$, in Lemma~\ref{lem:kraft_completion} respectively
 correspond to $0$, $k-2$, $S$, $W_k\cup L$, $A$   here),
since $0 < K(W_k \cup L) < 2^{-k} + 2^{-(k-1)} < 2^{-(k-2)} \le 1 = K(\Alphabet)$,
there exists a subset $A \subseteq \Alphabet - (W_k \cup L)$
such that
$K(A) = 2^{-(k-2)} - K(W_k \cup L)$.
Then 
\begin{align}
K(L \cup A) 
 &= K(W_k \cup L \cup A) - K(W_k)\\
 &= K(W_k \cup L) + K(A) - K(W_k)\\
 &= 2^{-(k-2)} - K(W_k) \\
 &= 2^{-(k-2)} - 2^{-k} \\
 &= 2^{-(k-1)} + 2^{-k}
\end{align}
so the Huffman-Kraft partition of $L \cup A$ provided by Lemma~\ref{lem:kraft_partition}
consists of $2$ disjoint subsets, $E$ and $F$,
of $L\cup A$ such that
$E \cup F = L \cup A$,
along with $K(E) = 2^{-k}$
and $K(F) = 2^{-(k-1)}$.
Then $P(E) > 0$,
and $P(F) \ge K(F)$ by \eqref{eq:7707}.

Since $L \subseteq L \cup A = E \cup F$,
we have $\Alphabet - (W_k \cup E \cup F) \subseteq \Alphabet - L \subseteq W \cup T$.
Therefore,
$P(\Alphabet - (W_k \cup E \cup F)) \ge K(\Alphabet - (W_k \cup E \cup F)) = 1-2^{-(k-2)}$
by \eqref{eq:208} and \eqref{eq:210}.
But we also have
\begin{align}
P(\Alphabet - (W_k \cup E \cup F))
&= 1 - P(W_k) - P(E) - P(F) \\
&< 1 - P(W_k) - P(F) \\
&\le 1 - 2K(W_k) - K(F) \label{eq:7708}\\
&= 1-2^{-(k-1)} - 2^{-(k-1)} \\
&= 1-2^{-(k-2)} \\
&= K(\Alphabet - (W_k \cup E \cup F)), 
\end{align}
which is a contradiction,
where
\eqref{eq:7708} follows from \eqref{eq:208}.
Therefore, our assumption was false that $K(L) < 2^{-(k-2)}$, so in fact
$K(L) \ge 2^{-(k-2)}$.
Thus, $k \ne 2$, for otherwise $K(L) \ge 1$ which contradicts $W\ne \emptyset$.
Therefore, $k \ge 3$.
Since $W_1 = \dots = W_{k-1} = \emptyset$,
we have $2K(W) < 2^{-(k-2)} \le K(L)$,
and so Lemma~\ref{lem:a_2b_inequality} shows
$P(W) < P(L)$.
\end{proof}

Theorem~\ref{thm:Huffman-beats-SF} guarantees that for each $n$, 
and for all non-dyadic sources of size $n$,
Huffman codes always strictly competitively dominate Shannon-Fano codes.
On the other hand,
we saw in Section~\ref{sec:not-optimal}
that as the source size grows,
an increasingly large fraction of non-dyadic sources 
have prefix codes that
strictly competitively dominate Huffman codes.
Said more casually,
Huffman codes usually are dominated by another code
but always dominate Shannon-Fano codes.

\clearpage

\section{Bound on competitive advantage over Huffman codes}
\label{sec:One-third-bounds}

\newcommand{\MinProb}[1]{\left(\displaystyle\min_{z\in #1}P(z)\right)}
\newcommand{\MinProbNoParens}[1]{{\displaystyle\min_{z\in #1}P(z)}}

In this section,
we derive an upper bound on the competitive advantage of an arbitrary
prefix code over a Huffman code for a given source,
and show that for every source of size at least four the upper bound can be 
approached arbitrarily closely by some sources.
We show that no prefix code can have a competitive advantage 
of $\frac{1}{3}$ or higher
over any Huffman code
(Theorem~\ref{thm:ub-one-third-Huffman}),
and in fact this upper bound is tight
in that it can be approached arbitrarily closely from below
for all source sizes, by at least some sources
(Theorem~\ref{thm:lb-one-third-Huffman}).

If $A$ is a subset of a source's alphabet, then at least one of the following
three Huffman-Kraft sum conditions is satisfied:
(i) $A$ is empty and $K(A)=0$;
(ii) $A$ is the entire alphabet and $K(A)=1$;
or
(iii) $K(A)\in (0,1)$ and is a finite sum of negative integer powers of $2$,
and therefore has a binary expansion of the form
$K(A) = 0.b_1 b_2 \ldots$,
with 
$b_i \in \{0,1\}$ for all $i$,
and where the number of nonzero bits is at least one and is finite.

If $A$ is a subset of a source alphabet,
then the probability $P(A)$ and the Huffman-Kraft sum $K(A)$ are related.
When $A$ is empty, $P(A)=K(A)=0$,
and when $A$ is the entire source alphabet, $P(A)=K(A)=1$.
If the source is dyadic, then $P(A)=K(A)$ is always true.
For non-dyadic sources, 
the relationship between $P(A)$ and $K(A)$ is more complicated.
We next establish some lemmas that are used to prove
Theorem~\ref{thm:ub-one-third-Huffman}.
Some of these lemmas
relate the probabilities and the 
Huffman-Kraft sums of source alphabet subsets.

In this section,
for any given source,
if $C$ is a prefix code that competes against a Huffman code for the same source,
then define the events 
$W$ (i.e., $C$ ``wins''),
$L$ (i.e., $C$ ``loses''),
and $T$ (i.e., $C$ ``ties''),
as in \eqref{eq:7262b}--\eqref{eq:7262}
(taking $C_1=C$, and $C_2$ as the Huffman code).

The following two lemmas establish consequences that occur when a prefix code has
a positive competitive advantage over a Huffman code.

\begin{lemma}
For any source,
if a prefix code $C$ has a 
positive competitive advantage over a Huffman code,
then for at least one source symbol,
$C$ produces a longer codeword than the Huffman codeword.
\label{lem:L_not_empty}
\end{lemma}

\newcommand{\ProofOfLemmaLNotEmpty}{
\begin{proof}[Proof of Lemma \ref{lem:L_not_empty}]
It suffices to show that $L$ is nonempty.
Since $C$ has a positive competitive advantage over a Huffman code $H$, we have
$P(W) > P(L)$.
Then,
\begin{align}
0 &\le
E[l_C(X)] - E[l_H(X)] \label{eq:7038}\\
&= \sum_{y\in W} \left( l_C(y) - l_H(y) \right) P(y)
 + \sum_{y\in L} \left( l_C(y) - l_H(y) \right) P(y) \label{eq:7040}\\
&< \sum_{y\in L} \left( l_C(y) - l_H(y) \right) P(y) \label{eq:7039}
\end{align}
where
\eqref{eq:7038} follows since
a prefix code $C$ cannot have a lower expected length than the Huffman code for a given source;
\eqref{eq:7040} follows since $l_C(y) - l_H(y) = 0$ for all $y\in T$;
and
\eqref{eq:7039} follows since
$l_C(y) - l_H(y) < 0$ for all $y\in W$
and $P(W)>P(L) \ge 0$, so $W\ne \emptyset$.
If $L = \emptyset$,
then \eqref{eq:7039} would yield a contradiction.
\end{proof}
}

%\ProofOfLemmaLNotEmpty

\begin{lemma}
For any source,
if a prefix code $C$
has a positive competitive advantage over a Huffman code,
then the Huffman-Kraft sums of the set $W$ of wins and set $L$ of losses of $C$ satisfy
$K(W) < K(L)$.
\label{lem:w_le_l}
\end{lemma}

\newcommand{\ProofOfLemmaWLeL}{
\begin{proof}[Proof of Lemma \ref{lem:w_le_l}]
Let $H$ denote the Huffman code and
suppose, to the contrary, that 
\begin{align}
K(W) \ge K(L). \label{eq:7020}
\end{align}
Let $\Alphabet$ be the source's alphabet
and let $T = \Alphabet - (W\cup L)$ be the set of ties.
Then
\begin{align}
1 
&\ge \sum_{x \in \Alphabet} 2^{-l_C(x)} \label{eq:7292}\\
&= \sum_{x \in W} 2^{-l_C(x)}
 + \sum_{x \in T} 2^{-l_C(x)}
 + \sum_{x \in L} 2^{-l_C(x)} \\
&> \sum_{x \in W} 2^{-l_C(x)}
 + \sum_{x \in T} 2^{-l_C(x)} \label{eq:7025}\\
&= \sum_{x \in W} 2^{-l_C(x)}
 + \sum_{x \in T} 2^{-l_H(x)} \label{eq:7023}\\
&\ge \sum_{x \in W} 2^{-l_H(x)+1}
   + \sum_{x \in T} 2^{-l_H(x)} \label{eq:7022}\\
&= 2K(W) + K(T) \\
&\ge K(W) + K(T) + K(L) \label{eq:7021}\\
&= 1, \label{eq:7024}
\end{align}
a contradiction,
where
\eqref{eq:7292} is the Kraft inequality (\ref{lem:Kraft-Inequality}) applied to
                the prefix code $C$;
\eqref{eq:7025} follows from $L \ne \emptyset$ by Lemma~\ref{lem:L_not_empty};
\eqref{eq:7023} follows from $l_C(x) = l_H(x)$ when $x\in T$;
\eqref{eq:7022} follows from $l_C(x) \le l_H(x)-1$ when $x\in W$;
\eqref{eq:7021} follows from \eqref{eq:7020};
and
\eqref{eq:7024} follows since the Huffman tree is complete.
\end{proof}
}

%\ProofOfLemmaWLeL

\newcommand{\ProofOfLemmaFracInc}{
\begin{lemma}
If $b > a$,
then $\frac{x+a}{x+b}$ is monotonically increasing in $x$ for all $x\ne -b$.
\label{lem:frac_inc}
\end{lemma}

\begin{proof}[Proof of Lemma \ref{lem:frac_inc}]
$\frac{d}{dx}\left(\frac{x+a}{x+b}\right) = \frac{b-a}{(x+b)^2} > 0$.
\end{proof}
}

% \ProofOfFracInc

If the Kraft sum of one set of Huffman codewords is 
less than or equal to that
of another set of Huffman codewords,
then the probability of the first set can be greater than that of the second set.
The following two lemmas provide limits on how great this
excess probability can be under certain conditions.

\begin{lemma}
Let $U$ and $V$ be subsets of a source alphabet
whose Huffman-Kraft sums satisfy
$K(U) = K(V) = 2^{-i}$ for some integer $i\ge 0$,
and such that $|U|\ge 2$.
Then $P(U) \le 2 P(V)$.
\label{lem:a_le_2b}
\end{lemma}

\newcommand{\ProofOfLemmaALeTwoB}{
\begin{proof}[Proof of Lemma \ref{lem:a_le_2b}]
By Lemma~\ref{lem:strongly-monotone},
Huffman codes are strongly monotone,
so Huffman-Kraft sums obey the strong monotonicity condition.

Let $u \in U$ be an element of minimum Huffman-Kraft sum
among the elements of $U$.
Since $|U| \ge 2$,
there exists $m \ge i+1$ such that
$K(u) = 2^{-m}$,
for otherwise 
$K(U) \ge 2 \cdot 2^{-i} > 2^{-i} = K(U)$,
a contradiction.
Therefore,
the binary expansion of $K(U - \{u\}) = K(U) - K(u) = 2^{-i} - 2^{-m}$
has $1$s in positions $i+1,\dots,m$
and $0$s in all other positions.
By Lemma~\ref{lem:kraft_partition},
there exists a Huffman-Kraft partition $U_1,\dots,U_m$ of $U-\{u\}$.
In particular,
$K(U_{i+1}) = 2^{-(i+1)}$,
and so $K(U - U_{i+1}) = K(U) - K(U_{i+1}) = 2^{-i} - 2^{-(i+1)} = 2^{-(i+1)}$.
Then by strong monotonicity
$P(U_{i+1}),P(U-U_{i+1}) \le P(V)$,
and so
$P(U) = P(U_{i+1}) + P(U-U_{i+1}) \le 2P(V)$.
\end{proof}
}

%\ProofOfLemmaALeTwoB

\begin{lemma}
Given a source with alphabet $\Alphabet$,
let $U$ and $V$
be disjoint subsets of $S$
with Huffman-Kraft sums satisfying
$K(U) < 2^{-i} \le K(V)$
for some integer $i \ge 0$.
Then $P(U) < 2P(V)$.
\label{lem:u_l_2v}
\end{lemma}

\newcommand{\ProofOfLemmaULTwoV}{
\begin{proof}[Proof of Lemma \ref{lem:u_l_2v}]
Since $K(V) \ge 2^{-i}$,
we have $P(V) > 0$.
If $K(U) = 0$
then $P(U) = 0 < 2P(V)$,
so we may assume $K(U) > 0$.

By Lemma~\ref{lem:kraft_partition},
there exists a Huffman-Kraft partition
$V_1,V_2,\dots$ of $V$.
Let $k \ge 0$ be the smallest integer
such that $V_k \neq \emptyset$.
Since $K(V) \ge 2^{-i}$,
we have $k \le i$.
By Lemma~\ref{lem:kraft_completion}
(taking $B=U'$, $U=S$, $A=U$, and $j=k$),
there exists a subset $U' \subseteq S - U$
such that
$K(U') = 2^{-k} - K(U) > 0$.
Then $|U \cup U'| \ge 2$
since both $U$ and $U'$ are nonempty,
and also $K(U \cup U') = 2^{-k} = K(V_k)$,
so $P(U \cup U') \le 2 P(V_k)$
by Lemma~\ref{lem:a_le_2b}
(taking $A=U\cup U'$, $B=V_k$, and $i=k$).
Therefore
\begin{align}
P(U)
&\le 2P(V_k) - P(U')
< 2P(V_k)
\le 2P(V).
\end{align}
\end{proof}
}

%\ProofOfLemmaULTwoV

For any Huffman-Kraft partition $A_1,A_2,\dots$
of a subset $A$ of a source alphabet and for any integer $k \ge 1$
we will define the notation
$A_{<k} = A_1 \cup \dots \cup A_{k-1}$
and
$A_{\le k} = A_1 \cup \dots \cup A_k$,
as well as
$A_{>k} = A_{k+1} \cup A_{k+2}\dots$
and
$A_{\ge k} = A_k \cup A_{k+1} \dots$.

The following is the key lemma used to prove 
Theorem~\ref{thm:ub-one-third-Huffman}.

\begin{lemma}
Given a source with alphabet $\Alphabet$,
suppose $U, V \subseteq \Alphabet$
are disjoint subsets
with Huffman-Kraft sums satisfying
$K(U) < K(V)$.
Then $P(U) - P(V) < \frac{1}{3}$.
\label{lem:u_v_onethird}
\end{lemma}

\newcommand{\ProofOfLemmaUVOneThird}{
\begin{proof}[Proof of Lemma \ref{lem:u_v_onethird}]
If $K(U) = 0$
then $P(U) - P(V) \le 0 < \frac{1}{3}$,
so suppose $K(U) > 0$.
By Lemma~\ref{lem:kraft_partition},
there exist Huffman-Kraft partitions
$U_1,U_2,\dots$ of $U$ and
$V_1,V_2,\dots$ of $V$.
Let $m \ge 0$ be the smallest integer
such that $V_m \neq \emptyset$
and $U_m = \emptyset$;
such an integer exists
because $K(U) < K(V)$.
Since
$P(U) - P(V) \le P(U) - P(V_{\le m})$,
and $U$ and $V_{\le m}$ are disjoint,
without loss of generality
we will assume $V_i = \emptyset$
for all $i > m$.

We now assert that for certain nonnegative integers $k\le m-1$,
there exists $A \subseteq \Alphabet$
satisfying the following three conditions:
\begin{align}
K(A) &= 2^{-k} \label{eq:7919a}\\
U_{\ge k+1} \cup V_{\ge k+1} &\subseteq A \label{eq:7919b}\\
P(U_{\ge k+1}) - P(V_{\ge k+1}) &< \frac{1}{3} P(A). \label{eq:7919c}
\end{align}
Specifically, we will first show that this assertion is true when $k=m-1$.
Then, we will show inductively 
that whenever the assertion is true for some positive $k$
it must also be true for some smaller nonnegative $k$.
We then will conclude that the assertion must be true for $k=0$.

Once we have established the assertion is true for $k=0$,
we can further infer that
\begin{align}
P(U) - P(V)
&= P(U_{\ge 1}) - P(V_{\ge 1})
< \frac{1}{3} P(A)
= \frac{1}{3}, 
\end{align}
where $P(A)=1$ since $K(A)=1$,
thus proving the lemma.

\textbf{Base Step: $k=m-1$}.
Since $K(\Alphabet - (U_{< m} \cup V_{< m})) = 1 - K(U_{< m} \cup V_{< m})$
is an integer multiple of $2^{-(m-1)}$,
and
\begin{align}
0 < K(V_m) \le K(U_{\ge m} \cup V_{\ge m}) = K(U_{> m}) + K(V_m) < 2^{-(m-1)},
\end{align}
Lemma~\ref{lem:kraft_completion} shows
(taking 
 $A\leftarrow U_{\ge m}\cup V_{\ge m}$, 
 $B\leftarrow U'$, 
 $U\leftarrow S-(U_{<m}\cup V_{<m})$, 
 $i=j\leftarrow m-1$ )
there exists 
$U' \in 
 = \Alphabet - (U \cup V)$ 
such that 
$K(U' \cup U_{\ge m} \cup V_{\ge m}) = 2^{-(m-1)}$.
Then setting
\begin{align}
A &= U' \cup U_{\ge m} \cup V_{\ge m} \label{eq:7903}
\end{align}
gives
\begin{align}
K(A) &= 2^{-(m-1)}. \label{eq:7901}
\end{align}
Also,
since $K(U_{\ge m}) = K(U_{>m}) < 2^{-m} = K(V_m) = K(V_{\ge m})$,
Lemma~\ref{lem:u_l_2v} shows $P(U_{\ge m}) < 2P(V_{\ge m})$.
Therefore,
\begin{align}
\frac{P(U_{\ge m}) - P(V_{\ge m})}{P(A)}
&< \frac{P(U_{\ge m}) - P(V_{\ge m})}{P(U_{\ge m}) + P(V_{\ge m})} \label{eq:7910} \\
&< \frac{P(V_{\ge m})}{3 P(V_{\ge m})} \label{eq:7911} \\
&= \frac{1}{3}, \label{eq:7902}
\end{align}
where
\eqref{eq:7910} follows from 
$P(U')>0$ since
$K(U') = K(A) - K(U_{\ge m} \cup V_{\ge m})
 > 2^{-(m-1)} - 2^{-(m-1)} = 0$;
and
\eqref{eq:7911} follows from 
$P(U_{\ge m}) < 2P(V_{\ge m})$ and Lemma~\ref{lem:frac_inc}.

By \eqref{eq:7901}, \eqref{eq:7903}, and \eqref{eq:7902},
these conditions hold for $k = m-1$.

\textbf{Inductive Step: $1 \le k\le m-1$}.
Notice that if $m=1$,
then the base case of $k=m-1$ automatically proves the assertion for $k=0$,
so we may assume $m\ge 2$.

Assume that 
\eqref{eq:7919a} -- \eqref{eq:7919c}
hold for some positive integer $k\le m-1$ and for some $A \subseteq S$.
We will show that there exists $A' \subseteq \Alphabet$
that satisfies the three conditions above
for some index $j \in \{0,\dots,k-1\}$.

From the definition of $m$,
we have $K(U_i) = K(V_i)$
for all $i < m$.
In particular,
$K(U_k) = K(V_k)$
since $k \le m-1$.
Then $K(U_k) + K(V_k) \in \{0,2^{-(k-1)}\}$,
so $K(U_{\le k} \cup V_{\le k}) = K(U_{< k} \cup V_{< k}) + K(U_k \cup V_k)$
is an integer multiple of $2^{-(k-1)}$.
Therefore,
$K(\Alphabet - U_{\le k} \cup V_{\le k}) = 1 - K(U_{\le k} \cup V_{\le k})$
is an integer multiple of $2^{-(k-1)}$.
Since $0 < K(A) = 2^{-k} < 2^{-(k-1)}$,
and $A \subseteq \Alphabet - U_{\le k} \cup V_{\le k}$,
Lemma~\ref{lem:kraft_completion} shows
(taking 
 $A\leftarrow A$, 
 $B\leftarrow B$, 
 $U\leftarrow S-(U_{\le m}\cup V_{\le m})$, 
 $i=j\leftarrow k-1$ )
there exists $B \subseteq \Alphabet - (A \cup U_{\le k} \cup V_{\le k})$
with $K(B) = 2^{-(k-1)} - K(A) = 2^{-k}$.

\textbf{Case 1: $K(U_k) = K(V_k) = 0$}. 
Let $A' = A \cup B$.
Then
$K(A') = K(A) + K(B) = 2^{-(k-1)}$
and
\begin{align}
U_{\ge k} \cup V_{\ge k} 
&= U_{\ge k+1} \cup V_{\ge k+1} \subseteq A \subseteq A'
\end{align}
from \eqref{eq:7919b}.
Also,
\begin{align}
P(U_{\ge k}) - P(V_{\ge k})
&= P(U_{\ge k+1}) - P(V_{\ge k+1})
< \frac{1}{3} P(A)
< \frac{1}{3} P(A')
\end{align}
from \eqref{eq:7919c}
and $K(A) = K(A') - K(B) < K(A')$.
Thus $A'$ satisfies conditions \eqref{eq:7919a} -- \eqref{eq:7919c}
except the index $k$ has been reduced to $j=k-1 \ge 0$.

\textbf{Case 2: $K(U_k) = K(V_k) = 2^{-k}$}. 
Let $j \le k$ be the smallest integer
such that $K(U_i) = 2^{-i} = K(V_i)$
for all $i \in \{j,\dots,k\}$.
Then $j \ge 2$,
since otherwise, $j = 1$
would imply $1 \ge K(U) + K(V) > 2K(U) \ge 2K(U_1) = 1$,
a contradiction.
Also, from the definition of $j$,
we have $K(U_{j-1}) = K(V_{j-1}) = 0$.
Let
\begin{align}
A' &= A \cup B \cup (U_{\ge j} - U_{> k}) \cup (V_{\ge j} - V_{> k}). \label{eq:7920}
\end{align}
Then
\begin{align}
U_{\ge j-1} \cup V_{\ge j-1}
&= U_{\ge j} \cup V_{\ge j} \\
&= (U_{\ge j} - U_{> k}) \cup (V_{\ge j} - V_{> k}) 
    \cup U_{\ge k+1} \cup V_{\ge k+1} \\
&\subseteq (U_{\ge j} - U_{> k}) \cup (V_{\ge j} - V_{> k}) \cup A 
              \label{eq:7930}\\
&\subseteq A' \label{eq:7905}
\end{align}
where 
\eqref{eq:7930} follows from \eqref{eq:7919b};
and
\begin{align}
K(A')
&= K(A) + K(B) + K(U_{\ge j} - U_{> k}) + K(V_{\ge j} - V_{> j}) \\
&= 2^{-k} + 2^{-k} + 2 \sum_{i=j}^k 2^{-i} \\
&= 2^{-(j-2)}. \label{eq:7906}
\end{align}

Now let $E \in \{A,B\}$ such that
$P(E) = \min(P(A),P(B))$.
Note $K(E) = 2^{-k}$,
since $K(A) = K(B) = 2^{-k}$.
Since
\begin{align}
K((V_{\ge j} - V_{> k}) \cup E)
&= \sum_{i=j}^k 2^{-i} + K(E)
= 2^{-(j-1)} - 2^{-k} + 2^{-k}
= 2^{-(j-1)}
> 2^{-j}
= K(U_j),
\end{align}
the strong monotonicity of Huffman codes by 
Lemma~\ref{lem:strongly-monotone}
implies
$P(U_j) \le P(V_{\ge j} - V_{> k}) + P(E)$.

Suppose $j < k$.
Then for all $i \in \{j+1,\dots,k\}$,
we have $P(U_i) \le P(V_{i-1})$
by strong monotonicity of Huffman-Kraft sums,
since $K(U_i) = 2^{-i} < 2^{-(i-1)} = K(V_{i-1})$.
Thus
\begin{align}
P(U_{\ge j+1} - U_{> k})
&= P(U_{j+1}) + \dots + P(U_k) \\
&\le P(V_j) + \dots + P(V_{k-1}) \\
&= P(V_{\ge j} - V_{> k-1}).
\end{align}
Therefore
\begin{align}
P(U_{\ge j} - U_{> k})
&= P(U_j) + P(U_{\ge j+1} - U_{> k}) \\
&\le P(V_{\ge j} - V_{> k}) + P(E) + P(V_{\ge j} - V_{> k-1}) \\
&= 2 P(V_{\ge j} - V_{> k-1}) + P(V_k) + P(E). \label{eq:7914}
\end{align}
On the other hand, suppose $j=k$.
Then $V_{\ge j} - V_{> k-1} = \emptyset$,
so we also have
\begin{align}
P(U_{\ge j} - U_{> k})
&= P(U_j) \\
&\le P(V_{\ge j} - V_{> k}) + P(E) \\
&= 2 P(V_{\ge j} - V_{> k-1}) + P(V_k) + P(E). \label{eq:7915}
\end{align}

Now we combine both the cases $j<k$ and $j=k$.
Since
$V_m \subseteq A$
and
$K(V_m) = 2^{-m} < 2^{-k} = K(A)$,
neither $V_m$ nor $A-V_m$ is empty,
so $|A| = |V_m| + |A - V_m| \ge 2$.
Then since $K(A) = 2^{-k} = K(V_k)$,
Lemma~\ref{lem:a_le_2b} shows $P(A) \le 2P(V_k)$.
Thus
\begin{align}
P(E) &= \min(P(A),P(B)) \le \min( 2P(V_k), P(B) ),
\end{align}
and so
\begin{align}
2P(V_k) + P(E) + P(B) \ge 3P(E). \label{eq:7916}
\end{align}
Finally,
we apply
Lemma~\ref{lem:frac_inc}
by using the values
\begin{align}
x &= P(U_{\ge j} - U_{> k})\\
a &= - P(V_{\ge j} - V_{> k})\\
b &= P(V_{\ge j} - V_{> k}) + P(B)\\
x' &= 2 P(V_{\ge j} - V_{> k-1}) + P(V_k) + P(E).
\end{align}
It is clear that $a<b$,
and we have $x \le x'$ from
\eqref{eq:7914} (for $j<k$) and \eqref{eq:7915} (for $j=k$),
Thus,
\begin{align}
\frac{x+a}{x+b} &\le \frac{x'+a}{x'+b}\\
\frac{P(U_{\ge j} - U_{> k}) - P(V_{\ge j} - V_{> k})}
     {P(U_{\ge j} - U_{> k}) + P(V_{\ge j} - V_{> k}) + P(B)}
&\le \frac{P(V_{\ge j} - V_{> k-1}) + P(E)}{3 P(V_{\ge j} -
             V_{> k-1}) + 2P(V_k) + P(E) + P(B)} \\
&\le \frac{P(V_{\ge j} - V_{> k-1}) + P(E)}{3 P(V_{\ge j} -
             V_{> k-1}) + 3 P(E)} \label{eq:7913} \\
&= \frac{1}{3}, \label{eq:7918}
\end{align}
where
\eqref{eq:7913} follows from \eqref{eq:7916}.
Therefore,
\begin{align}
P(U_{\ge j-1}) - P(V_{\ge j-1})
&= P(U_{\ge j}) - P(V_{\ge j}) \\
&= P(U_{\ge j} - U_{> k}) - P(V_{\ge j} - V_{> k}) + P(U_{\ge k+1}) - P(V_{\ge k+1}) \\
&< \frac{1}{3}(P(U_{\ge j} - U_{> k}) + P(V_{\ge j} - V_{> k}) + P(B) + P(A)) \label{eq:7917} \\
&= \frac{1}{3} P(A'), \label{eq:7907}
\end{align}
where
\eqref{eq:7917} follows from \eqref{eq:7918} and \eqref{eq:7919c};
and
\eqref{eq:7907} follows from \eqref{eq:7920}.
Thus $A'$ satisfies the three conditions
\eqref{eq:7919a} -- \eqref{eq:7919b}
by way of
\eqref{eq:7906}, \eqref{eq:7905}, and \eqref{eq:7907},
except the index $k$ has been reduced to $j-2 \in \{0,\dots,k-2\}$.
\end{proof}
}

%\ProofOfLemmaUVOneThird

\begin{theorem}
For any source,
the competitive advantage
of any prefix code over a Huffman code 
is less than $\frac{1}{3}$.
\label{thm:ub-one-third-Huffman}
\end{theorem}

\begin{proof}
Let $C$ denote an arbitrary prefix code for the source.
Let $W$ and $L$ denote the sets of wins and losses, respectively,
of $C$ over the Huffman code.
It suffices to assume the competitive advantage of $C$ 
over the Huffman code is positive,
so $W\ne\emptyset$.
Then Lemma~\ref{lem:L_not_empty} implies $L \ne \emptyset$,
and Lemma~\ref{lem:w_le_l} implies $K(W) < K(L)$.
Therefore,
$P(W) - P(L) < \frac{1}{3}$
by Lemma~\ref{lem:u_v_onethird}.
\end{proof}

% ---------------------------------------------------------------------------

The following theorem shows that for any size at least four,
sources can be found
whose competitive advantages over Huffman codes are arbitrarily 
close to $1/3$ and whose average lengths are arbitrarily close to that
of a Huffman code.

\begin{theorem}
For every $n\ge 4$,
there exists a source of size $n$ and a prefix code
that has a competitive advantage over a Huffman code 
arbitrarily close to $\frac{1}{3}$ and the code's average length 
is arbitrarily close to that of the Huffman code.
\label{thm:lb-one-third-Huffman}
\end{theorem}

\begin{proof}
Let $n\ge 4$
and $\epsilon > 0$,
and define
$\alpha = \frac{\epsilon/2}{1 - 2^{-n+3}}$.
Let the source be of size $n$ and with symbol probabilities:
\begin{align}
p_1 &= \frac{1}{3} + \epsilon\\
p_2 &= \frac{1}{3}\\
p_3 &= \frac{1}{3} - 2\epsilon\\
p_k &= \alpha 2^{4-k} \ \ \ \ \ \ (4 \le k \le n).
\end{align}
One can verify that $p_1 + \dots + p_n = 1$
and for each $k\in \{2, \dots, n\}$
we have $p_k > p_{k+1} + \dots + p_n$,
so the Huffman code for the source
assigns a word of length $k$ % i.e. $1^{k-1}0$ 
to $p_k$ for $k=1, \dots, n-1$,
and also a word of length $n-1$ % i.e. $1^{n-1}$ 
to $p_n$.

Define a prefix code $C$ which is identical to the Huffman code, 
except that it reassigns
$p_1$, $p_2$, and $p_3$ to codewords of lengths
$3$, $1$, and $2$, respectively.
The code $C$ will produce a shorter codeword than 
that of the Huffman code with probability $p_2+p_3$
and will produce a longer codeword with probability $p_1$.
Thus, the competitive advantage of $C$ over the Huffman code is 
$\Advantage = p_2+p_3 - p_1 = \frac{1}{3} - 3\epsilon$.

Denote the codeword lengths of the Huffman code by $l_i$.
The average length of the Huffman code is
\begin{align}
1\cdot\left(\frac{1}{3} + \epsilon\right) 
+ 2\cdot\left(\frac{1}{3}\right) 
+ 3\cdot\left(\frac{1}{3}-2\epsilon\right)
 + \displaystyle\sum_{k=4}^n p_i l_i
\end{align}
and the average length of $C$ is
\begin{align}
1\cdot\left(\frac{1}{3}\right) 
+ 2\cdot\left(\frac{1}{3}-2\epsilon\right) 
+ 3\cdot\left(\frac{1}{3}+\epsilon\right)
 + \displaystyle\sum_{k=4}^n p_i l_i
\end{align}
so their difference is
$-\frac{14}{3}\epsilon$.

In summary, 
the code $C$ achieves a competitive advantage over the Huffman code of
$\frac{1}{3} - 3\epsilon$
and has an average length at most
$\frac{14}{3}\epsilon$ greater than that of the Huffman code.
Taking $\epsilon$ arbitrarily small makes the 
competitive advantage approach $\frac{1}{3}$
and the average length difference approach zero.
\end{proof}

% ---------------------------------------------------------------------------

\clearpage

\section{Bound on competitive advantage over Shannon-Fano codes}
\label{sec:Shannon-Fano}

On one hand,
Shannon-Fano codes are efficient,
since they suffice in proving Shannon's source coding theorem
that says the average length of optimal block codes 
arbitrarily approaches from above the entropy of a source,
as the block size grows.
The proof uses the fact that the average length of 
a Shannon-Fano code is always 
within one bit of the source entropy,
and so the 
average length per symbol of a Shannon-Fano code for a source block of size $n$
is within $\frac{1}{n}$ bit of the source entropy.

One the other hand, 
Huffman codes are strictly better than Shannon-Fano codes 
in an average length sense for non-dyadic sources,
and perform equally well for for dyadic sources.
Similarly, in a competitive sense,
Theorem~\ref{thm:Huffman-beats-SF}
showed that 
Huffman codes strictly competitively dominate 
Shannon-Fano codes 
if and only if the source is not dyadic.

The competitive advantage of one code over 
a Shannon-Fano code (or, actually, any other code)
is trivially upper bounded by one,
and the average length of a code can be at most one bit less than that of a Shannon-Fano code.
The following theorem shows that 
there exist increasingly large sources 
with prefix codes that can approach both of these extremes
over Shannon-Fano codes simultaneously.

\begin{theorem}
For every positive integer $n$, 
there exists a source of size $n$ and a prefix code
that has a competitive advantage of at least 
$1 - 2^{-n+2}$ over a Shannon-Fano code for the source,
and the code's average length is 
at least 
$1 - 2^{-n+2}$ less than the average length of the Shannon-Fano code.
\label{thm:lb-one-Shannon-Fano}
\end{theorem}

\begin{proof}
Let $\epsilon \in \left( 0, 4^{-n} \right)$ and
let $X$ be a source of size $n$ whose probabilities are
\begin{align}
p_k &= 
 \begin{cases}
  2^{-k}    -      \epsilon & \text{if}\ 1 \le k \le n-1\\
  2^{-n+1} + (n-1)\epsilon & \text{if}\ k=n .
 \end{cases}
\end{align}
Since
$\ceil*{ \log_2 \frac{1}{p}} = m$
if and only if
$2^{-m} \le  p < 2^{-m+1}$,
a Shannon-Fano code for this distribution has codeword lengths
\begin{align}
l_k &= \left\lceil \log_2 \frac{1}{p_k} \right\rceil =
 \begin{cases}
   k+1 & \text{if}\ 1 \le k \le n-1\\
   n-1  & \text{if}\ k=n.
 \end{cases}
\end{align}
Note that $l_n$ was determined from the fact that
for all $n\ge 1$,
\begin{align}
(n-1)4^{-n} &< 2^{-n+1}.
\label{eq:950}
\end{align}

Let $C$ be a prefix code that assigns
the word $1^{k-1}0$
to the outcomes that have probability $p_k$ when $k<n$,
and assigns the word $1^{n-1}$ to the outcome with probability $p_n$.
This prefix code 
produces a shorter codeword than a Shannon-Fano code
whenever $1 \le k \le n-1$, and ties when $k=n$,
so its competitive advantage over a Shannon-Fano code is
lower bounded as
\begin{align}
\Advantage = 
1-p_n &= 1 - \frac{1}{2^{n-1}} - (n-1)\epsilon 
\ge 1 - \frac{1}{2^{n-1}} - \frac{n-1}{4^n} 
> 1 - \frac{1}{2^{n-2}}.\label{eq:123}
\end{align}
where \eqref{eq:123} follows from \eqref{eq:950}.

The difference between the average lengths of 
the Shannon-Fano code and the code $C$ is
\begin{align}
(n-1)p_n + \sum_{k=1}^{n-1} (k+1)p_k 
- (n-1)p_n - \sum_{k=1}^{n-1} k p_k 
&= \sum_{k=1}^{n-1} p_k
= 1 - p_n, \label{eq:124}
\end{align}
the same quantity as the competitive advantage previously computed
in \eqref{eq:123}.

\end{proof}

In the preceding proof,
the average length of the code $C$
is at most $2^{-n+2}$ more than the source entropy, since
\begin{align}
E[ l_C(X)) ] 
&< E[ l_{C_{SF}}(X)) ] - 1 + \frac{1}{2^{n-2}}\label{eq:120}\\
&< H(X) + 1 - 1 + \frac{1}{2^{n-2}}\label{eq:121}\\
&= H(X)+ \frac{1}{2^{n-2}}. 
\end{align}
where
\eqref{eq:120} follows from \eqref{eq:123} and \eqref{eq:124};
and
\eqref{eq:121} follows from Shannon's source coding theorem.
We also note that the term $2^{-n+2}$ that occurs in the bounds
of the theorem can be sharpened to be arbitrarily close to 
$2^{-n+1}$ but we chose to keep the proof simple instead.

% ---------------------------------------------------------------------------
\clearpage

\section{Small codes}
\label{sec:small-codes}

In this section,
we analyze which sources of size at most $4$
have competitively optimal Huffman codes.

\begin{theorem} 
Huffman codes are competitively optimal
for all sources of size at most $3$.
\label{thm:competitively-optimal-Huffman:n=3}
\end{theorem}

\begin{proof}
If the source is of size $1$ or $2$, the result is trivial,
so suppose the size is $3$.
Denote the source symbols by $1, 2, 3$ such that
$P(1) \ge P(2) \ge P(3) > 0$.
The word lengths of a Huffman code $H$ are $l_H(1)=1$ and 
$l_H(2)=l_H(3)=2$.
Let $C$ denote any other prefix code,
and use the notation $W$ and $L$, as in \eqref{eq:7262b}--\eqref{eq:7262}.
It is not possible for $1 \in W$,
since $l_C(1) \ge 1 = l_H(1)$.
If $2 \in W$, then $l_C(2)=1$ 
and therefore $l_C(1), l_C(3) \ge 2$, 
so $1 \in L$ and $3 \not\in W$,
which implies
the competitive advantage of $C$ over the Huffman code is 
$\Advantage = P(W)-P(L) = P(2) - P(L) \le P(2) - P(1) \le 0$.
Alternatively,
if $3 \in W$,
then we similarly conclude
$\Advantage \le 0$.
Finally, if $2,3 \not\in W$, then $P(W)=0$, so $\Advantage \le 0$.
\end{proof}

The following theorem determines for almost all sources of size $4$,
which ones have competitively optimal Huffman codes and which do not.
Since the $4$ probabilities $p_1, p_2, p_3, p_4$ in such a distribution
sum to $1$,
a characterization of which distributions lead to competitive optimality can be
described in terms of conditions on the $3$ quantities $p_1, p_2, p_3$.
Based on the Huffman construction process,
linear inequalities are obtained in these $3$ quantities, 
which in turn correspond to half-spaces in $\R^3$.
The intersection of these half-spaces is a convex polyhedron,
which turns out to be the union of two tetrahedra sharing a
common triangular face, that is, a hexahedron.

\begin{theorem}
Let $Q$ be the hexahedron with vertices 
$(\frac{1}{2}, \frac{1}{2}, 0)$,
$(\frac{2}{5}, \frac{1}{5}, \frac{1}{5})$,
$(\frac{1}{3}, \frac{1}{3}, \frac{1}{3})$,
$(\frac{1}{3}, \frac{1}{3}, \frac{1}{6})$,
$(\frac{1}{2}, \frac{1}{4}, \frac{1}{4})$.
For every source of size $4$
with probabilities 
$p_1 \ge p_2 \ge p_3 \ge p_4 > 0$,
a Huffman code is competitively optimal if the triple
$(p_1, p_2, p_3)$ 
lies in the exterior of $Q$,
and is not competitively optimal if the triple 
lies in the interior of $Q$.
\label{thm:competitively-optimal-Huffman:n=4}
\end{theorem}

\begin{proof}
Denote the source symbols by $1,2,3,4$ 
and their probabilities by $p_1, p_2, p_3, p_4$, respectively.
We will determine conditions on $p_1, \dots, p_4$
such that there exists a prefix code with a positive competitive
advantage over a Huffman code.
It suffices to consider complete prefix codes,
since any non-complete prefix code contains at least one codeword
that could be shortened without decreasing its competitive advantage.
The only possible codeword length distributions for such size-$4$
codes are $1,2,3,3$ and $2,2,2,2$.
In either case,
the Huffman algorithm merges the source symbols $3$ and $4$ 
to form a new symbol with probability
$p_3+p_4$.

Suppose $p_3+p_4 > p_1$.
Then the Huffman algorithm merges $1$ and $2$ 
and then the $(3, 4)$ symbol is merged with the $(1, 2)$ symbol 
to get a balanced tree
with codeword lengths $2,2,2,2$.
If a size-$4$ prefix code
achieves a positive competitive advantage over this Huffman code,
then it must have codeword lengths $1,2,3,3$,
for otherwise only ties would occur.
In this case,
the competitive advantage would be
the probability of the new code's length-$1$ word 
minus the sum of the probabilities of its two length-$3$ words,
which equals $p_1 - (p_3+p_4) < 0$,
so in fact the new code would be 
strictly competitively dominated by the Huffman code.
The competitive advantage would still not be positive even if
$p_3+p_4=p_1$ and the Huffman algorithm created codewords with lengths
$2,2,2,2$.

Alternatively, assume $p_3+p_4 \le p_1$
with the Huffman algorithm merging the $(3, 4)$ symbol with $2$,
and then merging the resulting $(2, (3, 4))$ symbol with $1$.
The resulting codeword lengths are $1,2,3,3$.
The competitive advantage of any depth-$2$ balanced tree
over the Huffman code would be
$p_3+p_4-p_1 \le 0$, 
so such codes are competitively dominated by the Huffman code.
Thus,
any code $C$ with a positive competitive advantage $\Advantage$
over the Huffman code
must have lengths $1,2,3,3$,
and hence $C$ just permutes the Huffman code's
assignment of codeword lengths to source symbols.

Suppose $l_C(1)=1$.
If $l_C(2)=2$, then $\Advantage=0$.
If $l_C(2)=3$ and $l_C(3)=2$, then $\Advantage = p_3-p_2 \le 0$.
If $l_C(2)=3$ and $l_C(4)=2$, then $\Advantage = p_4-p_2 \le 0$.

Alternatively, suppose $l_C(1)\ne 1$.
There are $9$ possible cases for
$(l_C(1), l_C(2), l_C(3), l_C(4))$:\\
(2,1,3,3): $\Advantage = p_2 - p_1 \le 0$ \\
(2,3,1,3): $\Advantage = p_3 - p_1 - p_2 \le 0$ \\ 
(2,3,3,1): $\Advantage = p_4 - p_1 - p_2 \le 0$ \\ 
(3,2,1,3): $\Advantage = p_3 - p_1 \le 0$ \\
(3,2,3,1): $\Advantage = p_4 - p_1 \le 0$ \\
(3,3,1,2): $\Advantage = p_3 + p_4 -  p_1 - p_2 \le 0$ \\
(3,3,2,1): $\Advantage = p_3 + p_4 -  p_1 - p_2 \le 0$ \\
(3,1,2,3): $\Advantage = p_2 + p_3 - p_1$ \\
(3,1,3,2): $\Advantage = p_2 + p_4 - p_1$. \\
So the only codes $C$ that can yield $\Advantage>0$ are the cases
$(3,1,2,3)$ and $(3,1,3,2)$.

Let us denote the following inequalities:
\begin{align}
&(I1):\ p_1 \ge p_2\\
&(I2):\ p_2 \ge p_3\\
&(I3):\ p_3 \ge p_4\\
&(I4):\ p_4 > 0\\
&(I5):\ p_3 + p_4 \le p_1\\  
&(I6):\ p_2 + p_3 > p_1\\
&(I7):\ p_2 + p_4 > p_1.
\end{align}

Inequalities (I1) -- (I6)
determine a set in $\R^3$ whose interior
is a hexahedron specified by the $5$ vertices
$(\frac{1}{2}, \frac{1}{2}, 0)$,
$(\frac{1}{3}, \frac{1}{3}, \frac{1}{3})$,
$(\frac{2}{5}, \frac{1}{5}, \frac{1}{5})$,
$(\frac{1}{3}, \frac{1}{3}, \frac{1}{6})$,
$(\frac{1}{2}, \frac{1}{4}, \frac{1}{4})$.
The first $3$ vertices satisfy $\sim$(I7) with equality,
the $4^{th}$ vertex satisfies (I7),
and $5^{th}$ vertex satisfies $\sim$(I7).
Therefore, 
the hexahedron is cut into two tetrahedra by (I7)
and is known as a triangular dipyramid.

The Huffman code is competitively optimal 
in the exterior of this hexahedron,
is not competitively optimal 
in the interior of this hexahedron,
and is sometimes competitively optimal on the boundary.
\end{proof}

\begin{corollary}
If a source of size $4$ is chosen uniformly at random from a flat Dirichlet
distribution, then the probability its Huffman code is competitively optimal
is $2/3$.
\label{cor:probability_n=4}
\end{corollary}

\begin{proof}
The hexahedron in
Theorem~\ref{thm:competitively-optimal-Huffman:n=4}
is a union of two tetrahedra,
whose volumes 
are computed using determinants as
(e.g., ~\cite{Newson-Polyhedron-Volume})
\begin{align}
\frac{1}{6} \cdot
 \begin{vmatrix}
  1/2 & 1/2 &  0  & 1\\
  2/5 & 1/5 & 1/5 & 1\\
  1/3 & 1/3 & 1/3 & 1\\
  1/3 & 1/3 & 1/6 & 1
 \end{vmatrix}
&= \frac{1}{6}\cdot \frac{1}{180}
\hspace*{2cm}
\frac{1}{6} \cdot
 \begin{vmatrix}
  1/2 & 1/2 &  0  & 1\\
  2/5 & 1/5 & 1/5 & 1\\
  1/3 & 1/3 & 1/3 & 1\\
  2/5 & 1/5 & 1/5 & 1
 \end{vmatrix}
= \frac{1}{6}\cdot \frac{1}{120} \ \ .
\end{align}
The set of all $p_1, p_2, p_3, p_4$
satisfying $p_1 \ge p_2 \ge p_3 \ge p_4 > 0$
and $p_1 + p_2 + p_3 + p_4 = 1$
is determined by the $4$ inequalities
\begin{align}
\ p_1 &\ge p_2\\              
\ p_2 &\ge p_3\\              
\ p_1 + p_2 + 2p_3 &\ge 1\\   
\ p_1 + p_2 + p_3 &< 1.
\end{align}
These form a tetrahedron with vertices
$(1,0,0)$,
$(\frac{1}{2}, \frac{1}{2}, 0)$,
$(\frac{1}{3}, \frac{1}{3}, \frac{1}{3})$,  
$(\frac{1}{4}, \frac{1}{4}, \frac{1}{4})$
whose volume is 
\begin{align}
\frac{1}{6} \cdot
 \begin{vmatrix}
   1  &  0  &  0  & 1\\
  1/2 & 1/2 &  0  & 1\\
  1/3 & 1/3 & 1/3 & 1\\
  1/4 & 1/4 & 1/4 & 1
 \end{vmatrix}
&= \frac{1}{6}\cdot \frac{1}{24}\ \ .
\end{align}
Thus the probability of randomly selecting a source from a flat Dirichlet distribution 
whose Huffman code is not competitively optimal is
$(\frac{1}{180} + \frac{1}{120}) / \frac{1}{24} = \frac{1}{3}$.
So the probability the Huffman code is competitively optimal is $\frac{2}{3}$.
\end{proof}

\clearpage

\section{Experimental evidence}
\label{sec:experimental-evidence}

We demonstrate numerically that
if a source is chosen at random,
then as the source size grows, the probability becomes nearly zero that 
a Huffman code will be competitively optimal.
Experimentally, 
this probability is less than $1\%$ when the source size is at least $20$.
That is, with near certainty 
each Huffman code will be competitively dominated by
some other prefix code as the source size increases.
This indicates that
for most sources,
from a competitive advantage point of view,
there really is no ``best'' code to use.
Each code can be strictly competitively dominated 
by another in never-ending cycles of code sequences.

One way to generate source probabilities $p_1, \dots, p_n$ 
chosen according to a flat Dirichlet distribution
is to choose $n$ points
independently and uniformly on a circle of circumference $1$
and then use the $n$ distances 
between neighboring points as the desired probabilities.
Such a procedure treats all sources equally
and indeed yields interesting results.

For any source of size $n$,
exhaustively checking whether each complete prefix code
competitively dominates the Huffman code
appears to become a computationally infeasible task as $n$ grows,
since the number of such prefix codes grows quickly.
However, 
Lemma~\ref{lem:LeafProbLessThanDifference}
gives a sufficient condition 
for a Huffman tree to not be competitively optimal,
which allows us to obtain a lower bound on the probability that a Huffman code
is not competitively optimal for a given source.
Thus we can randomly select many sources and determine if
such a condition holds, 
in which case we can then declare the Huffman code
not competitively optimal.
This suboptimal condition turns out to be overwhelmingly sufficient
to observe that 
the probability is practically zero that the
Huffman code of a randomly chosen source is competitively optimal
even for relatively small source sizes.

For each source size $n \in \{3, \dots, 34\}$,
we generated $10^6$ sources 
from a flat Dirichlet distribution,
i.e., chosen uniformly at random on the 
$(n-1)$-dimensional simplex embedded in $\R^n$.
For each such source we determined 
whether the sufficient condition of 
Lemma~\ref{lem:LeafProbLessThanDifference}
was satisfied.
Fig~\ref{fig:1} plots for each $n$
the fraction of the randomly generated sources that satisfied the 
sufficient condition.
That is, the true fraction of the randomly generated sources for which
a Huffman code was not competitively optimal lies above the plotted curve.
The observed lower bound curve quickly tends toward $1$,
so the true fraction of 
the randomly generated sources with
competitively non-optimal Huffman codes
tends toward $1$ as well.

For the case of $n=3$,
Theorem~\ref{thm:competitively-optimal-Huffman:n=3} 
guarantees that $100\%$ of the randomly 
chosen sources will have competitively optimal Huffman codes,
which is exactly what was observed experimentally.

For the case $n=4$,
Corollary~\ref{cor:probability_n=4} 
gives a $2/3$ probability of a randomly chosen
source to have a competitively optimal Huffman code.
The experimentally observed upper bound 
was $66.6992\%$.

For $n\ge 5$,
one can see that the probability a randomly chosen source has
a competitively optimal Huffman code
rapidly decreases towards $0$,
and in fact no such competitively optimal Huffman codes were
observed out of the million chosen for each
$n \ge 31$.

\Hide{
\begin{table}
\begin{center}
\begin{tabular}{|r|r|}
\hline
Source & Upper bound on the number of  \\
size   & competitively optimal Huffman code \\ 
       & out of $10^6$ chosen randomly sources\\ \hline
 3 & 1,000,000 \\ \hline  
 4 &  666,992 \\ \hline  
 5 &  532,397 \\ \hline  
 6 &  412,334 \\ \hline  
 7 &  318,936 \\ \hline  
 8 &  236,188 \\ \hline  
 9 &  164,586 \\ \hline  
10 &  110,338 \\ \hline  
11 &  70,707  \\ \hline  
12 &  44,453  \\ \hline  
13 &  27,746  \\ \hline  
14 &  16,896  \\ \hline  
15 &  10,174  \\ \hline  
16 &  5,850   \\ \hline  
17 &  3,568   \\ \hline  
18 &  2,106   \\ \hline  
19 &  1,177   \\ \hline  
20 &  643    \\ \hline  
21 &  360    \\ \hline  
22 &  210    \\ \hline  
23 &  121    \\ \hline  
24 &  81     \\ \hline  
25 &  40     \\ \hline  
26 &  17     \\ \hline  
27 &  16     \\ \hline  
28 &  6      \\ \hline  
29 &  2      \\ \hline  
30 &  3      \\ \hline  
31 &  0      \\ \hline  
32 &  0      \\ \hline  
33 &  0      \\ \hline  
34 &  0      \\ \hline  
\end{tabular}
\end{center}

\caption{
For each source size, $10^6$ sources were 
chosen uniformly at random from a simplex
(i.e., using a Dirichlet distribution).
The table indicates an upper bound 
(using Lemma~\ref{lem:LeafProbLessThanDifference})
on how many of these source
have Huffman codes which are competitively optimal.
}
\label{tab:1}
\end{table}
}

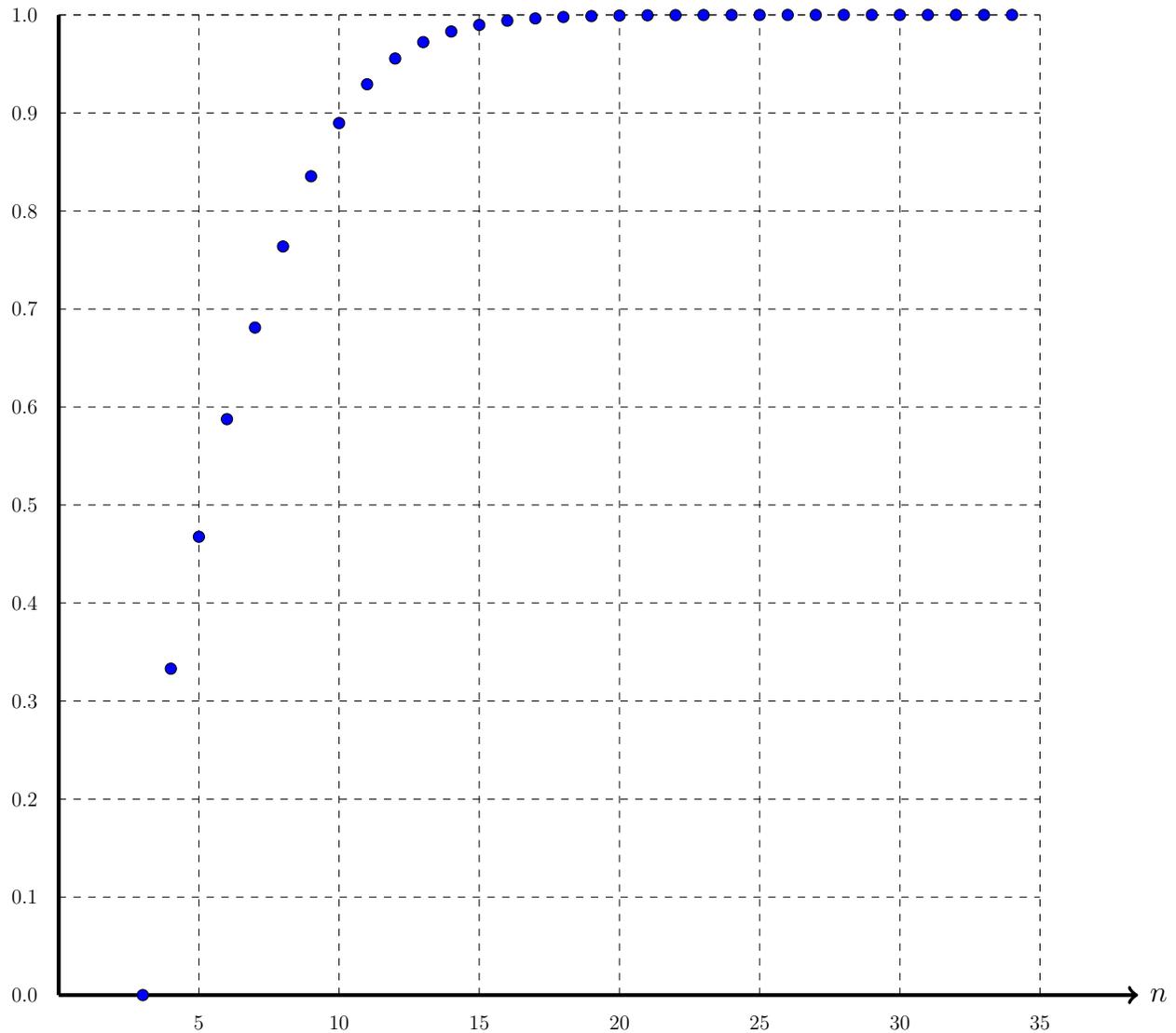
\begin{figure}
\newcommand{\Inflate}{35}
\begin{center}
\begin{tikzpicture}[scale=0.4]
\draw[->, ultra thick]  (0,0) -- ( {\Inflate * 1.1},0) node[right] {$n$};
\draw[-, ultra thick]  (0,0) -- ( 0,\Inflate) node[above] {};
\draw[-, ultra thin, dashed]  (0,1 * \Inflate)--( 35, 1 * \Inflate);
\foreach \i in {1,...,7} {
  \tikzmath{
    int \k;
    \k = 5 * \i;
  }
  \draw[-, ultra thin, dashed]  (\k,0)--(\k, \Inflate);
  \node[black, below=2mm, scale=0.7] at (\k, 0) {$\k$}; }

\foreach \i in {0,...,10} {
  \tikzmath{\j = \i / 10;}

\draw[-, ultra thin, dashed]  (0,\j*\Inflate)--(35,\j*\Inflate);

\pgfkeys{/pgf/number format/.cd, fixed, zerofill, precision=1}
  \node[black, left=2mm, scale=0.7] at (0, {\j*\Inflate}) 
       {$\pgfmathprintnumber{\j}$}; 
}
\draw [fill=blue] ( 3 , {  0.0      * \Inflate }) circle (2.0mm); 
\draw [fill=blue] ( 4 , {  0.333008 * \Inflate }) circle (2.0mm); 
\draw [fill=blue] ( 5 , {  0.467603 * \Inflate }) circle (2.0mm); 
\draw [fill=blue] ( 6 , {  0.587666 * \Inflate }) circle (2.0mm); 
\draw [fill=blue] ( 7 , {  0.681064 * \Inflate }) circle (2.0mm); 
\draw [fill=blue] ( 8 , {  0.763812 * \Inflate }) circle (2.0mm); 
\draw [fill=blue] ( 9 , {  0.835414 * \Inflate }) circle (2.0mm); 
\draw [fill=blue] (10 , {  0.889662 * \Inflate }) circle (2.0mm); 
\draw [fill=blue] (11 , {  0.929293 * \Inflate }) circle (2.0mm); 
\draw [fill=blue] (12 , {  0.955547 * \Inflate }) circle (2.0mm); 
\draw [fill=blue] (13 , {  0.972254 * \Inflate }) circle (2.0mm); 
\draw [fill=blue] (14 , {  0.983104 * \Inflate }) circle (2.0mm); 
\draw [fill=blue] (15 , {  0.989826 * \Inflate }) circle (2.0mm); 
\draw [fill=blue] (16 , {  0.994150 * \Inflate }) circle (2.0mm); 
\draw [fill=blue] (17 , {  0.996432 * \Inflate }) circle (2.0mm); 
\draw [fill=blue] (18 , {  0.997894 * \Inflate }) circle (2.0mm); 
\draw [fill=blue] (19 , {  0.998823 * \Inflate }) circle (2.0mm); 
\draw [fill=blue] (20 , {  0.999357 * \Inflate }) circle (2.0mm); 
\draw [fill=blue] (21 , {  0.999640 * \Inflate }) circle (2.0mm); 
\draw [fill=blue] (22 , {  0.999790 * \Inflate }) circle (2.0mm); 
\draw [fill=blue] (23 , {  0.999879 * \Inflate }) circle (2.0mm); 
\draw [fill=blue] (24 , {  0.999919 * \Inflate }) circle (2.0mm); 
\draw [fill=blue] (25 , {  0.999960 * \Inflate }) circle (2.0mm); 
\draw [fill=blue] (26 , {  0.999983 * \Inflate }) circle (2.0mm); 
\draw [fill=blue] (27 , {  0.999984 * \Inflate }) circle (2.0mm); 
\draw [fill=blue] (28 , {  0.999994 * \Inflate }) circle (2.0mm); 
\draw [fill=blue] (29 , {  0.999998 * \Inflate }) circle (2.0mm); 
\draw [fill=blue] (30 , {  0.999997 * \Inflate }) circle (2.0mm); 
\draw [fill=blue] (31 , {  1.000000 * \Inflate }) circle (2.0mm); 
\draw [fill=blue] (32 , {  1.000000 * \Inflate }) circle (2.0mm); 
\draw [fill=blue] (33 , {  1.000000 * \Inflate }) circle (2.0mm); 
\draw [fill=blue] (34 , {  1.000000 * \Inflate }) circle (2.0mm); 
\end{tikzpicture}
\end{center}
\caption{Lower bound on the fraction of $10^6$ randomly chosen sources
whose Huffman code is not competitively optimal,
as a function of the source size $n$.
For $n=15$ Huffman codewords,
about $99\%$ of randomly selected sources did not have
competitively optimal Huffman codes.
For $n \ge 31$, all $10^6$ randomly chosen sources had Huffman codes
that were not competitively optimal.
}
\label{fig:1}
\end{figure}
% ---------------------------------------------------------------------------

\clearpage

\section{Discussion}
\label{sec:discussion}

It is well known that for any uniquely decodable code (UD),
there exists a prefix code that is length equivalent.
Thus, for any two UD codes there exist two prefix codes such that the
competitive advantage of one UD code over the other is the same
as the competitive advantage of the corresponding prefix code over the other.
As a consequence, 
while all of the results in this paper have been stated for prefix codes,
they equivalently hold for UD codes as well.

The results presented in this paper indicate that for any given source,
the largest possible competitive advantage of a prefix code 
over a Huffman code lies in the interval $[0,1/3)$,
and that for at least some sources the upper limit can be arbitrarily approached.
One interesting question for future study might be to determine the typical 
distribution of such maximal competitive advantages over Huffman codes for
a source randomly chosen from, say, a flat Dirichlet distribution.

It would be interesting to determine if some of the results presented
in this paper can be generalized to $D$-ary alphabets when $D>2$.
It seems plausible that similar statements hold for
Theorems~\ref{thm:Huffman-not-competitively-optimal-in-limit},
\ref{thm:Huffman-beats-SF}, and
\ref{thm:lb-one-Shannon-Fano},
and also for
Theorems~\ref{thm:ub-one-third-Huffman} and
\ref{thm:lb-one-third-Huffman}
with the $1/3$ bound replaced by $\frac{1}{2D-1}$.

Some of our results are asymptotic as the source size $n$ grows.
It would indeed be interesting to understand the convergence rates in such situations.
For example,
how fast does the probability converge to zero that a randomly chosen
source has a competitively optimal Huffman code?

%===================================================
\clearpage

\begin{appendices}

\section{Proofs of lemmas}
\label{sec:appendixA}

\ProofOfLemmaKraftExistPrefixCode
\ProofOfLemmaKraftSumProbDisagree
\ProofOfLemmaLeafProbLessThanDifference
\ProofOfLemmaKraftPartition
\ProofOfLemmaKraftCompletion
\ProofOfLemmaAtwoBInequality
\ProofOfLemmaLNotEmpty
\ProofOfLemmaWLeL
\ProofOfLemmaFracInc
\ProofOfLemmaALeTwoB
\ProofOfLemmaULTwoV
\ProofOfLemmaUVOneThird

\end{appendices}

\textbf{Acknowledgment}:
The authors thank UCSD undergraduate student Marco Bazzani
for some helpful discussions.


\clearpage

\begin{thebibliography}{100}
%
\label{references}
\renewcommand{\baselinestretch}{0.9}
\setstretch{0.9}

\bibitem{Banks-golf-book}
Banks, R. 
\textit{How to Play Golf: A Guide to Learn the Golf Rules, Etiquette, Clubs, Balls, Types of Play, \& A Practice Schedule},
Italy, CRB Publishing, 2018.

\bibitem{Bell-Cover-1980}
R. Bell and T. M. Cover,
``Competitive optimality of logarithmic investment'',
\textit{Mathematical Operations Research}, 
vol. 5, no. 2, pp. 161 -- 166, May 1980.

\bibitem{Berstel-Perrin-Reutenauer-book-2009}
J. Berstel, D. Perrin, and C. Reutenauer,
\textit{Codes and Automata},
Encyclopedia of Mathematics and its Applications,
Cambridge University Press, 2009.

\bibitem{Bhatnagar-2021}
J.R. Bhatnagar, 
``Competitive optimality: A novel application in evaluating practical AI Systems'',
\textit{Engineering Applications of Artificial Intelligence}, 
vol. 102, article 104241, June 2021.

\bibitem{Billingsley-book-1986}
P. Billingsley,
\textit{Probability and Measure},
John Wiley \& Sons, New York, 1986.

\bibitem{CoZe-strongly-monotone-ArXiv}
S. Congero and K. Zeger,
``A characterization of optimal prefix codes'',
\textit{SIAM Journal on Discrete Mathematics},
(submitted on April 9, 2024).
Also available at: arXiv:2311.07007 [cs.IT].

\bibitem{Cover-1991}
T. M. Cover,
``On the competitive optimality of Huffman codes'',
\textit{IEEE Transactions on Information Theory},
vol. 37, no. 1, pp. 172 -- 174, January 1991.

\bibitem{Cover-Thomas-book-2006}
T. M. Cover and J. A. Thomas,
\textit{Elements of Information Theory},
2nd edition, New Jersey,
Wiley-Interscience, 2006.

\bibitem{Devroye-book}
L. Devroye,
\textit{Non-Uniform Random Variate Generation},
Springer-Verlag, New York, 1986.

\bibitem{Feder-1992}
M. Feder,
``A note on the competitive optimality of Huffman codes'',
\textit{IEEE Transactions on Information Theory},
vol. 38, no. 2, pp. 436 -- 439, March 1992.

\bibitem{Gallager-IT-1978} 
R. G. Gallager,
``Variations on a theme by Huffman",
\textit{IEEE Transactions on Information Theory},
vol. 24, no. 6, pp. 668 -- 8674, November 1978.

\bibitem{Gill-Wu-20-questions}
J. Gill and W. Wu,
``Twenty questions games always end with yes'',
ArXiv:1002.4907, February 2010.

\bibitem{Khosravifard-Saidi-Esmaeili-Gulliver-2007}
M. Khosravifard, H. Saidi, M. Esmaeili, and T. A. Gulliver, 
``The minimum average code for finite memoryless monotone sources'',
\textit{IEEE Transactions on Information Theory}, 
vol. 53, no. 3, pp. 955 -- 975, March 2007.

\bibitem{Linder-Tarokh-Zeger}
T. Linder, V. Tarokh, K. Zeger,
``Existence of optimal prefix codes for infinite source alphabets'',
\textit{IEEE Transactions on Information Theory}, 
vol. 43, no. 6, pp. 2026 -- 2028, November 1997.

\bibitem{Manickman-2019}
S. K. Manickam,
``Probability mass functions for which sources have
the maximum minimum expected length'',
\textit{National Conference on Communications (NCC)},
Bangalore, India, pp. 1 -- 6, February 20-23, 2019.

\bibitem{Newson-Polyhedron-Volume}
H. B. Newson,
``Volume of a polyhedron'',
\textit{Annals of Mathematics},
vol. 1, no. 1/4, pp. 108 -- 110, September 1899.

\bibitem{Dirichlet-book}
K. W. Ng, G.-L. Tian, and M.-L. Tang,
\textit{Dirichlet and Related Distributions: Theory, Methods and Applications},
John Wiley \& Sons, 2011.

\bibitem{Rastegari-Khosravifard-Narimani-Gulliver-2014}
P. Rastegari, M. Khosravifard, H. Narimani, and T. A. Gulliver,
``On the structure of the minimum average redundancy code for monotone sources'',
\textit{IEEE Communications Letters},
vol. 18, no. 4, pp. 664 -- 667, April 2014.

\bibitem{Yamamoto-Itoh-1995} 
H. Yamamoto and T. Itoh,
``Competitive optimality of source codes'',
\textit{IEEE Transactions on Information Theory},
vol. 41, no. 6, pp. 2015 -- 2019, November 1995.

\bibitem{Yamamoto-Yokoo-1995} 
H. Yamamoto and H. Yokoo,
``Average-sense optimality and competitive optimality 
  for almost instantaneous VF codes'',
\textit{IEEE Transactions on Information Theory},
vol. 47, no. 6, pp. 2174 -- 2184, September 2001.

\end{thebibliography}
\end{document}